\documentclass[runningheads]{llncs}

\usepackage{float,siunitx} 
\usepackage{amsmath,amssymb,amsfonts}
\usepackage{booktabs,caption,subcaption}
\usepackage{algorithmicx,setspace}
\usepackage[noend]{algpseudocode}
\usepackage{graphics,graphicx,xcolor}
\usepackage[utf8]{inputenc}
\usepackage{authblk}

\usepackage{pst-node}
\usepackage{multido}

\usepackage{varioref,hyperref}
\hypersetup{
  linkcolor  = red,
  citecolor  = green,
  urlcolor   = blue,
  colorlinks = true,
}
\usepackage[nameinlink,capitalize]{cleveref}

\usepackage{tikz}
\usetikzlibrary{positioning,chains,fit,shapes,calc}
\usepackage[most]{tcolorbox}
\usepackage{comment}
\usepackage[section]{placeins}

\usepackage{mathtools}

\newcommand{\glabel}[1]{\{X^b\}}

\DeclareRobustCommand\encircle[1]{%
  \tikz[baseline=(X.base)] 
    \node (X) [draw, shape=circle, inner sep=0, fill=black, text=white] {\strut #1};%
}

\title{The Effect of False Positives:\\Why Fuzzy Message Detection Leads to\\ Fuzzy Privacy Guarantees?}
\titlerunning{Why Fuzzy Message Detection Leads to Fuzzy Privacy Guarantees?}
\author{István András Seres\inst{1}\and Balázs Pejó\inst{2}\and Péter Burcsi\inst{1}}
\institute{Eötvös Loránd University, Budapest, Hungary\and CrySyS Lab, HIT/VIK/BME, Hungary}
\date{\today}

\authorrunning{I.A. Seres et al.}

\begin{document}

\maketitle

\begin{abstract}
    Fuzzy Message Detection (FMD) is a recent cryptographic primitive invented by Beck et al. (CCS'21) where an untrusted server performs coarse message filtering for its clients in a recipient-anonymous way. 
    In FMD --- besides the true positive messages --- the clients download from the server their cover messages determined by their false-positive detection rates.
    What is more, within FMD, the server cannot distinguish between genuine and cover traffic.
    In this paper, we formally analyze the privacy guarantees of FMD from three different angles. 
    
    First, we analyze three privacy provisions offered by FMD: recipient unlinkability, relationship anonymity, and temporal detection ambiguity.
    Second, we perform a differential privacy analysis and coin a relaxed definition to capture the privacy guarantees FMD yields. 
    Finally, we simulate FMD on real-world communication data.
    Our theoretical and empirical results assist FMD users in adequately selecting their false-positive detection rates for various applications with given privacy requirements. 
    
\keywords{Fuzzy Message Detection \and unlinkability \and anonymity \and differential privacy \and game theory}

\end{abstract}

\section{Introduction} \label{sec:introduction}

Fuzzy Message Detection (FMD)~\cite{beckfuzzy} is a promising, very recent privacy-enhancing cryptographic  primitive that aims to provide several desired privacy properties such as recipients' anonymity. In recipient-anonymous communication systems, not even the intended recipients can tell which messages have been sent to them without decrypting all messages. The main practical drawback for the users in a recipient-anonymous scheme such as messaging and payment systems is to \emph{efficiently and privately} detect the incoming messages or transactions. Decrypting all traffic in the system leads to a private but inevitably inefficient and bandwidth-wasting scan. This challenge is tackled by FMD, which allows the users to outsource the detection of their incoming traffic to an untrusted server in an efficient and privacy-enhanced way. It is assumed that messages/transactions are posted continuously to a potentially public board, e.g., to a permissionless public blockchain. It is expected that users are intermittently connected and resource-constrained. In the FMD scheme, whenever users are online, they download their genuine transactions as well as false-positive transactions according to their custom-set false-positive detection rate. The cryptographic technique behind FMD guarantees that true and false-positive messages are indistinguishable from the server's point of view. Thus, the false-positive messages act as cover traffic for genuine messages.

The FMD protocol caught the attention of many practitioners and privacy advocates due to the protocol's applicability in numerous scenarios. In general, it supports privacy-preserving retrieval of incoming traffic from store-and-forward delivery systems. We highlight two applications currently being implemented by multiple teams and waiting to be deployed in several projects~\cite{difrancesco2021umbra,devalence2021penumbra,lewis2021niwl,rondelet2021zeth}.

\begin{itemize}
    \item \textbf{Anonymous messaging. }
    In a recipient-anonymous messaging application, the senders post their recipient-anonymous messages to a store-and-forward server. If the server employs FMD, recipients can detect their incoming (and false-positive) messages in an efficient and privacy-enhanced way. Recently, the Niwl messaging application was deployed utilizing FMD~\cite{lewis2021niwl}.
    \item \textbf{Privacy-preserving cryptocurrencies \& stealth payments. }
    In privacy-preserving cryptocurrencies, e.g., Monero~\cite{noether2015ring}, Zcash~\cite{sasson2014zerocash}, or in a privacy-enhancing overlay, payment recipients wish to detect their incoming payments without scanning the whole ledger. At the time of writing, several privacy-enhancing overlays for Ethereum (e.g., Zeth~\cite{rondelet2021zeth}, Umbra~\cite{difrancesco2021umbra}) as well as for standalone cryptocurrencies (e.g., Penumbra~\cite{devalence2021penumbra}) are actively exploring the possibility of applying FMD in their protocols.
\end{itemize}

\paragraph{Contributions. } 
Despite the rapid adoption and interest in the FMD protocol, as far as we know, there is no study analyzing the provided privacy guarantees. Consequently, it is essential to understand the privacy implications of FMD. Furthermore, it is an open question how users need to choose their false-positive detection rates to achieve an efficiency-privacy trade-off suitable for their scenario.
In this work, we make the following contributions.

\begin{itemize}
    \item \textbf{Information-Theoretical Analysis. } 
    We assess and quantify the privacy guarantees of FMD and the enhanced $k$-anonymity it provides in the context of anonymous communication systems. We focus on three notions of privacy and anonymity: relationship anonymity, recipient unlinkability, and temporal detection ambiguity. We demonstrate that FMD does not provide relationship anonymity when the server knows the senders' identity. 
    What is more, we also study relationship anonymity from a game-theoretic point of view, and show that in our simplified model at the Nash Equilibrium the users do not employ any cover traffic due to their selfishness.
    Concerning recipient unlinkability and temporal detection ambiguity, we show that they are only provided in a meaningful way when the system has numerous users and users apply considerable false-positive detection rates.
    
    \item \textbf{Differential Privacy Analysis.}
    We adopt differential privacy (DP)~\cite{dwork2006differential} for the FMD scenario and coin a new definition, called Personalized Existing Edge Differential Privacy (PEEDP). Moreover, we analyze the number of incoming messages of a user with $(\varepsilon,\delta)$-differential privacy. The uncovered trade-off between the FMD's false-positive rates and DP's parameters could help the users to determine the appropriate regimes of false-positive rates, which corresponds to the level of tolerated privacy leakage.

    \item \textbf{Simulation of FMD on Real-World Data.}
    We quantitatively evaluate the privacy guarantees of FMD through open-source simulations on real-world communication systems. We show that the untrusted server can effortlessly recover a large portion of the social graph of the communicating users, i.e., the server can break relationship anonymity for numerous users.
\end{itemize}

\paragraph{Outline. }
In Section~\ref{sec:fmd}, we provide some background on FMD, while in Section~\ref{sec:systemmodel}, we introduce our system and threat model. 
In Section~\ref{sec:privacyguarantees}, we analyze the privacy guarantees of FMD while in Section~\ref{sec:dp} we study FMD using differential privacy.
In Section~\ref{sec:experiments}, we conduct simulations on real-world communication networks, and finally, in Section~\ref{sec:conclusion}, we conclude the paper.
\section{Fuzzy Message Detection}\label{sec:fmd}

The FMD protocol seeks to provide a reasonable privacy-efficiency trade-off in use cases where \emph{recipient anonymity} needs to be protected. Users generate detection keys and send them along with their public keys to the untrusted server. Senders encrypt their messages with their recipient's public key and create flag ciphertexts using the intended recipient's public key. Detection keys allow the server to test whether a flag ciphertext gives a match for a user's public key. If yes, the server stores the message for that user identified by its public key. In particular, matched flag ciphertexts can be false-positive matches, i.e., the user cannot decrypt some matched ciphertexts. Users can decrease their false-positive rate by sending more detection keys to the server. Above all, the FMD protocol ensures \emph{correctness}; whenever a user comes online, they can retrieve their genuine messages. The \emph{fuzziness} property enforces that each other flag ciphertext is tested to be a match approximately with probability $p$ set by the recipient. 
 
Besides recipient anonymity, FMD also aims to satisfy \emph{detection ambiguity}, which requires that the server cannot distinguish between true and false-positive matching flag ciphertexts provided that ciphertexts and detection keys are honestly generated. Hence, whenever a user downloads its matched messages, false-positive messages serve as cover traffic for the genuine messages. For formal security and privacy definitions of FMD and concrete instantiations, we refer the reader to Appendix~\ref{sec:app_fmd} and ultimately to~\cite{beckfuzzy}.
To improve readability, in Table~\ref{tab:notations}, we present the variables utilized in the paper: we refer to the downloaded (genuine or cover) flag ciphertext as a fuzzy tag. 

\begin{table}[tb]
    \centering
    \begin{tabular}{c|l}
       Variable & Description \\
       \midrule
       $U$ & Number of honest users (i.e., recipients and senders) \\
       $M$ & Number of all messages sent by honest users \\
       $p(u)$ & False-positive detection rate of recipient $u$ \\
       $tag(u)$ & Number of fuzzy tags received by $u$ (i.e., genuine and false positive) \\
       $tag_{v}(u)$ & Number of fuzzy tags received by $u$ from $v$ \\
       $in(u)$ & Number of genuine  incoming messages of $u$ \\
       $out(u)$ & Number of sent messages of $u$\\
    \end{tabular}
    \caption{Notations used throughout the paper.}
    \label{tab:notations}
\end{table}

\paragraph{Privacy-efficiency trade-off. }

If user $u$'s false-positive rate is $p(u)$, it received $in(u)$ messages and the total number of messages in the system is $M$, then the server will store $tag(u)\approx in(u)+p(u)(M-in(u))$ messages for $u$. Clearly, as the number of messages stored by the server increases, so does the strength of the anonymity protection of a message.
Note the trade-off between privacy and bandwidth efficiency: larger false-positive rate $p(u)$ corresponds to stronger privacy guarantees but also to higher bandwidth as more messages need to be downloaded from the server.\footnote{Similar scenario was studied in~\cite{bianchi2012better} concerning Bloom filters.} Substantial bandwidth can be prohibitive in certain use cases, e.g., for resource-constrained clients.
Even though in the original work of Beck et al.~\cite{beckfuzzy} their FMD instantiations support a restricted subset of $[2^{-l}]_{l\in\mathbb{Z}}$ as false-positive rates, in our privacy analysis, we lift these restrictions and assume that FMD supports any false-positive rate $p\in[0,1]$.

\paragraph{Provided privacy protection. }
The anonymity protection of FMD falls under the  ``hide-in-the-crowd'' umbrella term as legitimate messages are concealed amongst cover ones. More precisely, each legitimate message enjoys an enhanced version of the well-known notion of $k$-anonymity~\cite{sweeney2002k}.\footnote{Note that Beck et al. coined this as \emph{dynamic $k$-anonymity}, yet, we believe it does not capture all the aspects of their improvement. Hence, we renamed it with a more generic term. }
In more detail, the anonymity guarantee of the FMD scheme is essentially a ``dynamic'', ``personalized'', and ``probabilistic'' extension of $k$-anonymity. 
It is dynamic because $k$ could change over time as the overall number of messages could grow.
It is personalized because $k$ might differ from user to user as each user could set their own cover detection rates differently. 
Finally, it is probabilistic because achieved $k$ may vary message-wise for a user due to the randomness of the amount of selected fuzzy messages. 

To the best of our knowledge, as of today, there has not been a formal anonymity analysis of the ``enhanced $k$-anonymity'' achieved by the FMD protocol. 
Yet, there is already a great line of research demonstrating the weaknesses and the brittle privacy guarantees achieved by $k$-anonymity~\cite{domingo2008critique,machanavajjhala2007diversity}. Intuitively, one might hope that enhanced $k$-anonymity could yield strong(er) privacy and anonymity guarantees. However, we show both theoretically and empirically and by using several tools that this enhanced $k$-anonymity fails to satisfy standard anonymity notions used in the anonymous communication literature.\footnote{
For an initial empirical anonymity analysis, we refer the reader to the simulator developed by Sarah Jamie Lewis~\cite{lewis2021simulator}.} 
\section{System and Threat Model}\label{sec:systemmodel}

\paragraph{System Model. } 
In a typical application where the FMD scheme is applied, we distinguish between the following four types of system components where the users can simultaneously be  senders and recipients.

\begin{enumerate}
    \item \textbf{Senders:} They send encrypted messages to a message board. Messages are key-private, i.e., no party other than the intended recipient can tell which public key was used to encrypt the message. Additionally, senders post flag ciphertexts associated with the messages to an untrusted server. The goal of the flag ciphertexts is to allow the server and the recipients to detect their messages in a privacy-enhanced manner.
    \item \textbf{Message Board:} It is a database that contains the senders' messages. In many applications (e.g., stealth payments), we think of the message board as a public bulletin board; i.e., everyone can read and write the board. It might be implemented as a blockchain or as a centrally managed database, e.g., as would be the case in a messaging application. In either case, we assume that the message board is always available and that its integrity is guaranteed.
    \item \textbf{Server:} It stores the detection keys of recipients. Additionally, it receives and stores flag ciphertexts from senders and tests the flag ciphertexts with the recipient's detection keys. It forwards matching flag ciphertexts and their associated data (messages, transactions, etc.) to recipients whenever they query it. Typically, flag ciphertexts match numerous recipients' public keys.\footnote{
    In this work, we stipulate that a single server filters the messages for all users, i.e., a single server knows all the recipients' detection keys. }
    \item \textbf{Recipients:} The recipient obtains matching flag ciphertexts from the server. An application-dependent message is attached as associated data to each flag ciphertext, e.g., e-mail, payment data, or instant message. The number of matching ciphertexts is proportional to the recipient's false-positive detection rate and all the messages stored by the untrusted server.
\end{enumerate}

\paragraph{Threat model.}
Our focus is on the privacy and anonymity guarantees provided by FMD. Hence, we assume that the FMD scheme is a secure cryptographic primitive, i.e., the cryptographic properties of FMD (correctness, fuzziness, and detection ambiguity) hold. Senders and recipients are assumed to be honest. Otherwise, they can be excluded from the messages' anonymity sets. We consider two types of computationally-bounded attackers that can compromise user's privacy in an FMD scheme. The adversaries' goal is to learn as much information as possible about the relationship between senders, recipients, and messages. 

\begin{itemize}
    \item \textbf{Server:} Naturally, the server providing the FMD service can endanger the user's privacy since it has continuous access to every relevant information related to message detection. Specifically, the server knows the users' false-positive rates. It can analyze each incoming message, flag ciphertext, and their corresponding anonymity sets.
    \begin{itemize}
        \item \emph{Sender-oracle.}
        The server may know the sender of each message, i.e., a sender-oracle might be available in FMD. For instance, it is mandatory for the untrusted server if it only serves registered users. 
        We assumed solely in Section~\ref{sec:relanon} that such sender-oracle is available. 
        If FMD is integrated into a system where senders apply anonymous communication (e.g., use Tor to send messages and flag ciphertexts to the server), then sender-oracle is not accessible to the FMD server. 
    \end{itemize}
    \item \textbf{Eavesdropper:} A local passive adversary might observe the amount of data each user downloads from the server. Specifically, an eavesdropper could inspect the number of flag ciphertexts each user has received. Even though this attacker type does not have continual intrusion to the server's internal state, as we will show, it can still substantially decrease FMD user's privacy, e.g., if $p(u)$ is known, then the number of genuine incoming messages of users does not enjoy sufficiently high privacy protection (see Section~\ref{sec:dp}).
\end{itemize}

\section{Privacy Guarantees in FMD}\label{sec:privacyguarantees}

In this section, we analyze and quantify various privacy and anonymity guarantees provided by the FMD scheme. 
Specifically, in Sections~\ref{sec:reclink},~\ref{sec:relanon}, and ~\ref{sec:detambiguity}, we measure recipient unlinkability, assess relationship anonymity, and estimate detection ambiguity, respectively. 
Note that for the latter two property we provide experimental evaluations in Section~\ref{sec:experiments}, and we formulate a game in Appendix~\ref{sec:gametheory} concerning relationship anonymity.
We denote the security parameter with $\lambda$, and if an (probabilistic) algorithm $A$ outputs $x$, then we write $A\xrightarrow[]{}x$ ($A\xrightarrow[]{\$}x$). 
The Binomial distribution with success probability $p$ and number of trials $n$ is denoted as $\mathsf{Binom}(n,p)$, while a normal distribution with mean $\mu$ and variance $\sigma^2$ is denoted as $\mathcal{N}(\mu,\sigma^2)$. 

\subsection{Recipient Unlinkability} \label{sec:reclink}

In anonymous communication systems, recipient unlinkability is the cornerstone of anonymity guarantees. It ensures that it is hard to distinguish whether two different messages were sent to the same recipient or different ones. 
Whenever recipient unlinkability is not attained, it facilitates possibly devastating passive attacks, i.e., if an adversary can infer which messages are sent to the same recipient, then the adversary can effortlessly launch intersection attacks, see Appendix~\ref{sec:attacks}.
In the absence of recipient unlinkability, it is also possible to efficiently map every message to its genuine recipient by 1) clustering the messages that are sent to the same recipient and 2) see the intersection of the users who downloaded the flag ciphertexts sent to the same recipient.

We consider a definition of recipient unlinkability similar to the one introduced in~\cite{backes2013anoa}. Informally, in the recipient unlinkablity game, we examine two recipients $u_0,u_1$ and a sender $u_2$. The challenger $C$ generates uniformly at random $c\xleftarrow[]{\$}\{0,1\}$ and instructs $u_2$ to send message $m_\alpha$ to $u_c$. Afterwards, $C$ draws a uniformly random bit $b\xleftarrow[]{\$}\{0,1\}$. If $b=0$, then instructs $u_2$ to send a message $m_\beta$ to $u_c$. Otherwise $u_2$ sends $m_\beta$ to $u_{1-c}$. Adversary $\mathcal{A}$ examines the network, the flag ciphertexts and all communications and outputs $b'$ indicating whether the two messages were sent to the same recipient.

\begin{definition}[Recipient unlinkability (RU)]
An anonymous communication protocol $\Pi$ satisfies recipient unlinkability if for all probabilistic polynomial-time adversaries $\mathcal{A}$ there is a negligible function $\mathsf{negl}(\cdot)$ such that
\begin{equation}
    \Pr[\mathcal{G}^{RU}_{\mathcal{A},\Pi}(\lambda)=1]\leq\frac{1}{2}+\mathsf{negl}(\lambda),
\end{equation}
where the privacy game $\mathcal{G}^{RU}_{\mathcal{A},\Pi}(\lambda)$ is defined in Figure~\ref{fig:rugame} in Appendix~\ref{sec:APPPRI}.
\end{definition}

We denote the set of users who downloaded message $m$ by $\mathit{fuzzy}(m)$, i.e., they form the anonymity set of the message $m$. 
We estimate the advantage of the following adversary $\mathcal{A}$ in the $\mathcal{G}^{RU}_{\mathcal{A},\Pi}(\lambda)$ game: $\mathcal{A}$ outputs $0$ if $\mathit{fuzzy}(m_\alpha)\cap \mathit{fuzzy}(m_\beta)\neq\emptyset$ and outputs $1$ otherwise.
Note that $\mathcal{A}$ always wins if the same recipient was chosen by the challenger (i.e., $b=0$) because it is guaranteed by the correctness of the FMD scheme that $u_c\in \mathit{fuzzy}(m_\alpha)\cap \mathit{fuzzy}(m_\beta)$. Therefore, we have that $\Pr[\mathcal{G}^{RU}_{\mathcal{A},\Pi}(\lambda)=1\vert b=0]=1$. If two different recipients were chosen by the challenger in the $\mathcal{G}^{RU}_{\mathcal{A},\Pi}(\lambda)$ game (i.e., $b=1$), then $\mathcal{A}$ wins iff. $\mathit{fuzzy}(m_\alpha)\cap \mathit{fuzzy}(m_\beta)=\emptyset$. The advantage of the adversary can be computed as follows.

\begin{equation}
\label{eq:RUadvantage}
\begin{split}
    \Pr[\mathcal{G}^{RU}_{\mathcal{A},\Pi}(\lambda)=1\vert b=1] =
    \Pr[\cap_{m\in\{m_\alpha,m_\beta\}} fuzzy(m)=\emptyset\vert b=1] = \\
    = \sum^{U}_{i=1} \sum_{\substack{V\subseteq U \\
                                    \mid V\mid = i\\
                                    \alpha \in V}} 
        \Pr[\mathit{fuzzy}(m_{\alpha}) = V] \cdot
        \Pr[(V\cap \mathit{fuzzy}(m_{\beta}) = \emptyset] =\\
      = \sum^{U}_{i=1} \sum_{\substack{V\subseteq U \\
                                    \mid V\mid = i\\
                                    \alpha \in V}} 
        \left(\prod_{u_l \in V\setminus\{u_0\}} p(u_l)\cdot
        \prod_{u_l \in U\setminus V} (1-p(u_l))\right)\cdot
        \left(\prod_{u_l \in V} (1-p(u_l))\right) .
\end{split}
\end{equation}

We simplify the adversarial advantage in Equation~\ref{eq:RUadvantage} by assuming that $\forall i:p(u_i)=p$ and that the sizes of the anonymity sets are fixed at $\lfloor pU\rfloor$. Moreover, computer-aided calculations show that the following birthday paradox-like quantity can be used as a sufficiently tight lower bound\footnote{This lower bound is practically tight since the probability distribution of the adversary's advantage is concentrated around the mean $\lfloor pU\rfloor$ anyway.} for the recipient unlinkability adversarial advantage, whenever $p(u_i)$ are close to each other. 

\begin{equation}
    \label{eq:RUadvantagesimplified}
    \begin{split}
    \prod_{j=1}^{\lfloor pU\rfloor}\frac{U-\lfloor pU\rfloor-j}{U}=\prod_{j=1}^{\lfloor pU\rfloor}\Big(1-\frac{\lfloor pU\rfloor+j}{U}\Big)\approx\prod_{j=1}^{\lfloor pU\rfloor}e^{-\frac{\lfloor pU\rfloor+j}{U}}=\\
    =e^{-\frac{\sum_{j}(\lfloor pU\rfloor +j)}{U}}=e^{-\frac{3\lfloor pU\rfloor^2+\lfloor pU\rfloor}{2U}}\leq \Pr[\mathcal{G}^{RU}_{\mathcal{A},\Pi}(\lambda)=1\vert b=1].
    \end{split}
\end{equation}

The approximation is obtained by applying the first-order Taylor-series approximation for $e^x\approx1+x$, whenever $\vert x\vert\ll 1$. We observe that the lower bound for the adversary's advantage in the recipient unlinkability game is a negligible function in $U$ for a fixed false-positive detection rate $p$. Thus, in theory, the number of recipients $U$ should be large in order to achieve recipient unlinkability asymptotically. 
In practice, the classical birthday-paradox argument shows us that the two anonymity sets intersect with constant probability if $p=\mathcal{\theta}\Big(\frac{1}{\sqrt{U}}\Big)$. 
Our results suggest that a deployment of the FMD scheme should \emph{concurrently} have a large number of users with high false-positive rates in order to provide recipient unlinkability, see Figure~\ref{fig:anonanalysissub2} for the concrete values of Equation~\ref{eq:RUadvantagesimplified}.

\begin{figure}
\centering
\begin{subfigure}{.49\textwidth}
  \centering
  \includegraphics[width=\linewidth]{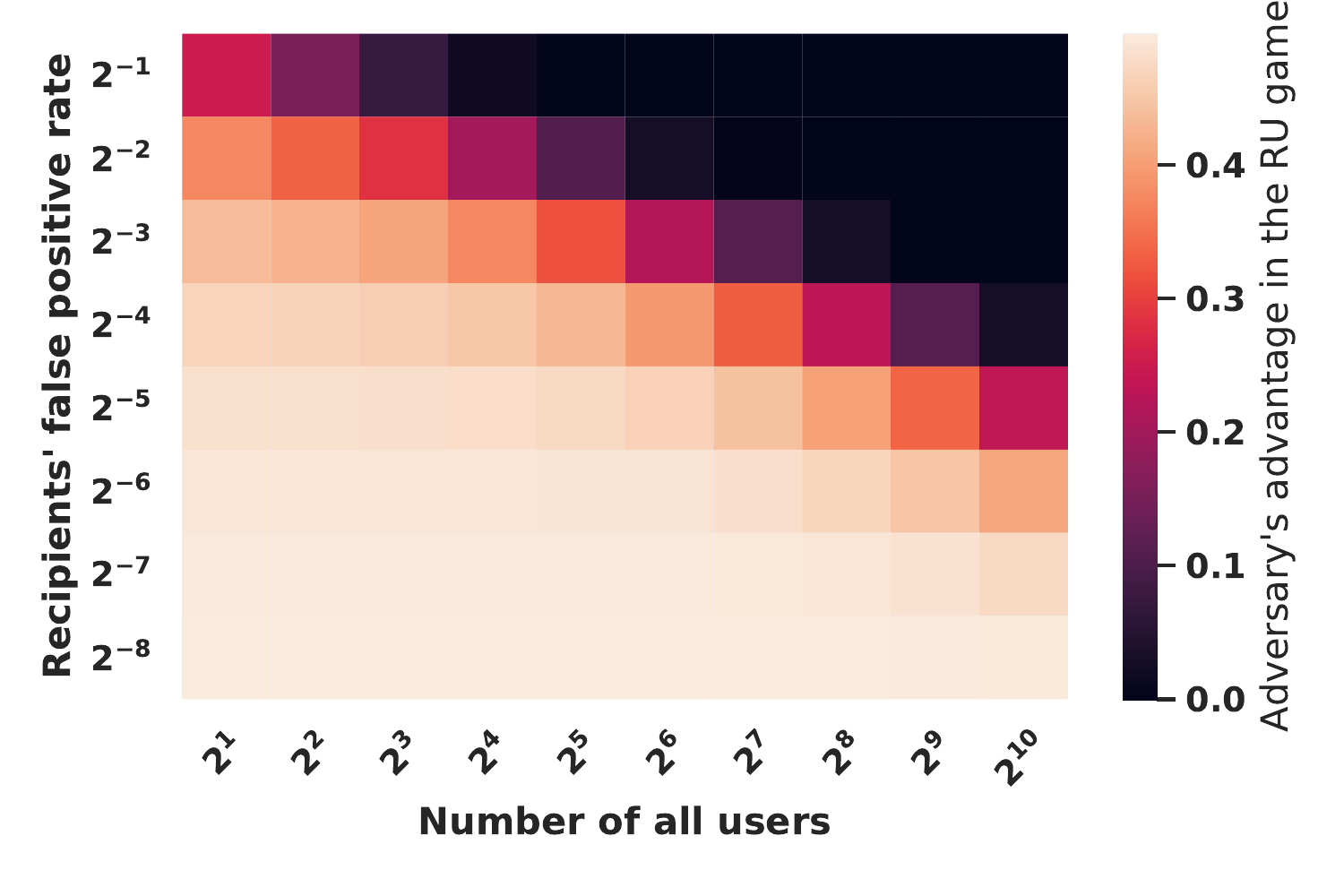}
  \caption{Approximate values of the recipient unlinkability adversarial advantage in the $\mathcal{G}_{RU}^{\mathcal{A},\Pi}(\lambda)$ game according to Equation~\ref{eq:RUadvantagesimplified}.}
  \label{fig:anonanalysissub2}
\end{subfigure}
\hfill
\begin{subfigure}{.49\textwidth}
  \centering
  \includegraphics[width=\linewidth]{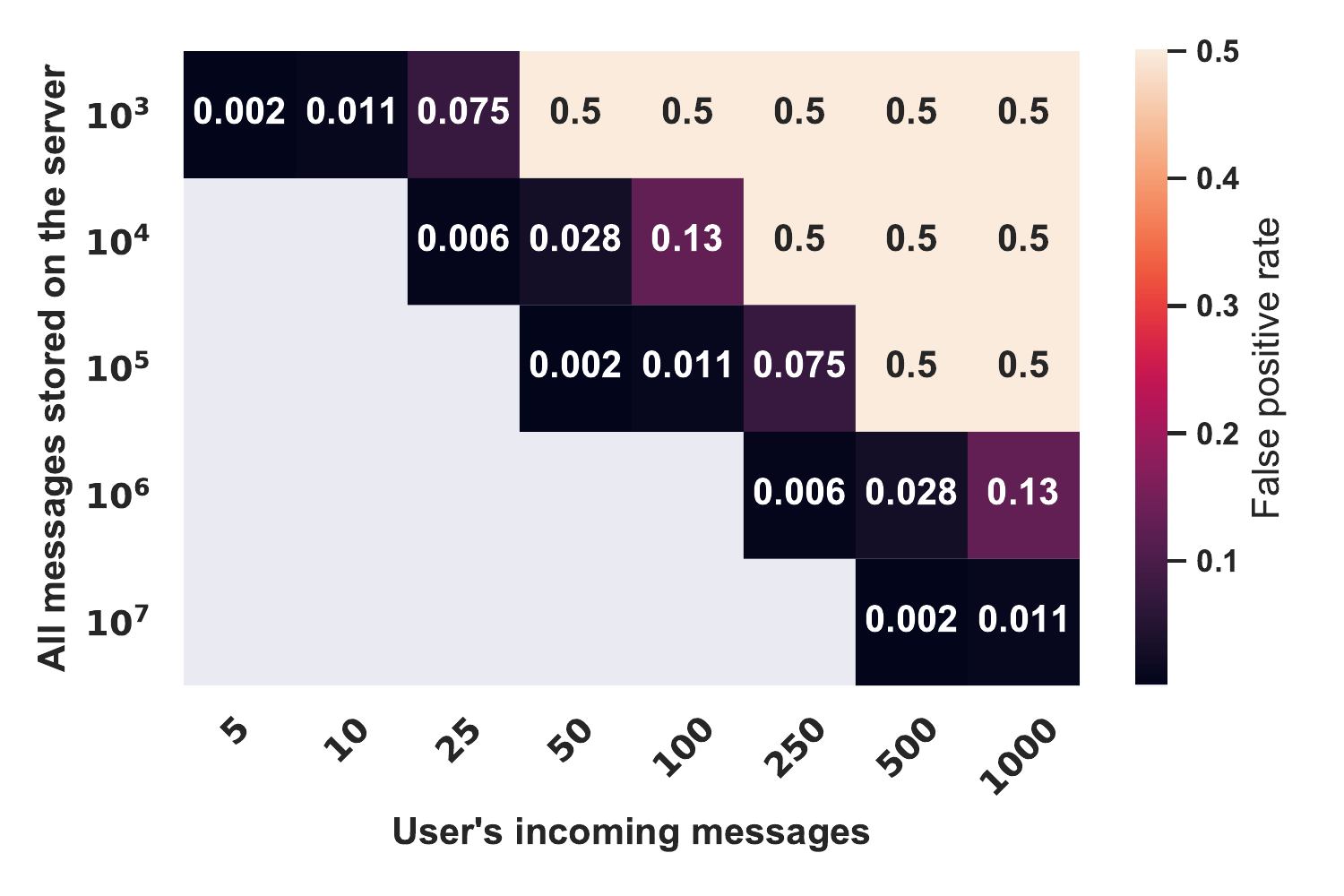}
  \caption{Smallest false-positive rates to obtain temporal detection ambiguity against the utilized statistical tests.}
  \label{fig:detectionambiguity}
\end{subfigure}
\caption{Recipient unlinkability and temporal detection ambiguity guarantees provided by the FMD scheme for various parameter settings. }
\label{fig:anonanalysis}
\end{figure}

\subsection{Relationship anonymity} \label{sec:relanon}

Relationship anonymity ensures that the adversary cannot determine the sender and the recipient of a communication at the same time. Intuitively, recipients applying low false-positive rates receive only a handful of fuzzy tags from peers they are not communicating with. Therefore, multiple fuzzy tags between a sender and a recipient can eradicate their relationship anonymity, given that the server knows the number of messages a sender issued. We assume the server knows the sender of each message, which holds whenever the untrusted server has access to a sender-oracle, see Section~\ref{sec:systemmodel}.

The number of fuzzy tags between a \emph{non-communicating pair} of reciever $u_1$ and sender $u_2$ follows $\mathsf{Binom}(\mathsf{out}(u_2),p(u_1))$. If  $tag_{u_2}(u_1)$ is saliently far from the expected mean $\mathsf{out}(u_2)p(u_1)$, then the untrusted server can deduce with high confidence that a relationship exists between the two users.
We approximate the binomial distribution above with $\mathcal{N}^{*}:=\mathcal{N}(\mathsf{out}(u_2) p(u_1),\mathsf{out}(u_2) p(u_1)(1-p(u_1)))$\footnote{Note that this approximation is generally considered to be tight enough when $\mathsf{out}(u_2)p(u_1)\geq5$ and $\mathsf{out}(u_2)(1-p(u_1))\geq5$. } so we can apply $Z$-tests to determine whether $u_2$ and $u_1$ had exchanged messages.
Concretely, we apply two-tailed $Z$-test\footnote{For senders with only a few sent messages ($\mathsf{out}(u_2)\leq 30$), one can apply $t$-tests instead of $Z$-tests.} for the hypothesis $H: tag_{u_2}(u_1)$ $\sim$ $\mathcal{N}(\mathsf{out}(u_2) p(u_1),\mathsf{out}(u_2) p(u_1)(1-p(u_1)))$. If the hypothesis is rejected, then users $u_2$ and $u_1$ are deemed to have exchanged messages.

Each recipient $u_1$ downloads on average $tag_{u_2}(u_1)\approx p(u_1)(\mathit{out}(u_2)-\mathit{in}_{u_2}(u_1))+\mathit{in}_{u_2}(u_1)$ fuzzy messages from the messages sent by $u_2$, where $\mathit{in}_{u_2}(u_1)$ denotes the number of genuine messages sent by $u_2$ to $u_1$. We statistically test with Z-tests (when $100\leq\mathsf{out}(u_2)$) and t-tests (when $\mathsf{out}(u_2)\leq30$) whether $tag_{u_2}(u_1)$ could have been drawn from the $\mathcal{N}^{*}$ distribution, i.e., there are no exchanged messages between $u_1$ and $u_2$. The minimum number of genuine messages $\mathit{in}_{u_2}(u_1)$ that statistically expose the communication relationship between $u_1$ and $u_2$ is shown for various scenarios in Figure~\ref{fig:relanonymity2}. We observe that the relationship anonymity of any pair of users could be broken by a handful of exchanged messages. This severely limits the applicability of the FMD scheme in use cases such as instant messaging. 
To have a meaningful level of relationship anonymity with their communicating peer, users should either apply substantial false-positive rates, or the server must not be able to learn the sender's identity of each message.
The latter could be achieved, for instance, if senders apply an anonymous communication system to send messages or by using short-lived provisional pseudo IDs where no user would send more than one message. 

\begin{figure}[t]
\centering
\begin{subfigure}{.49\textwidth}
  \centering
  \includegraphics[width=\linewidth]{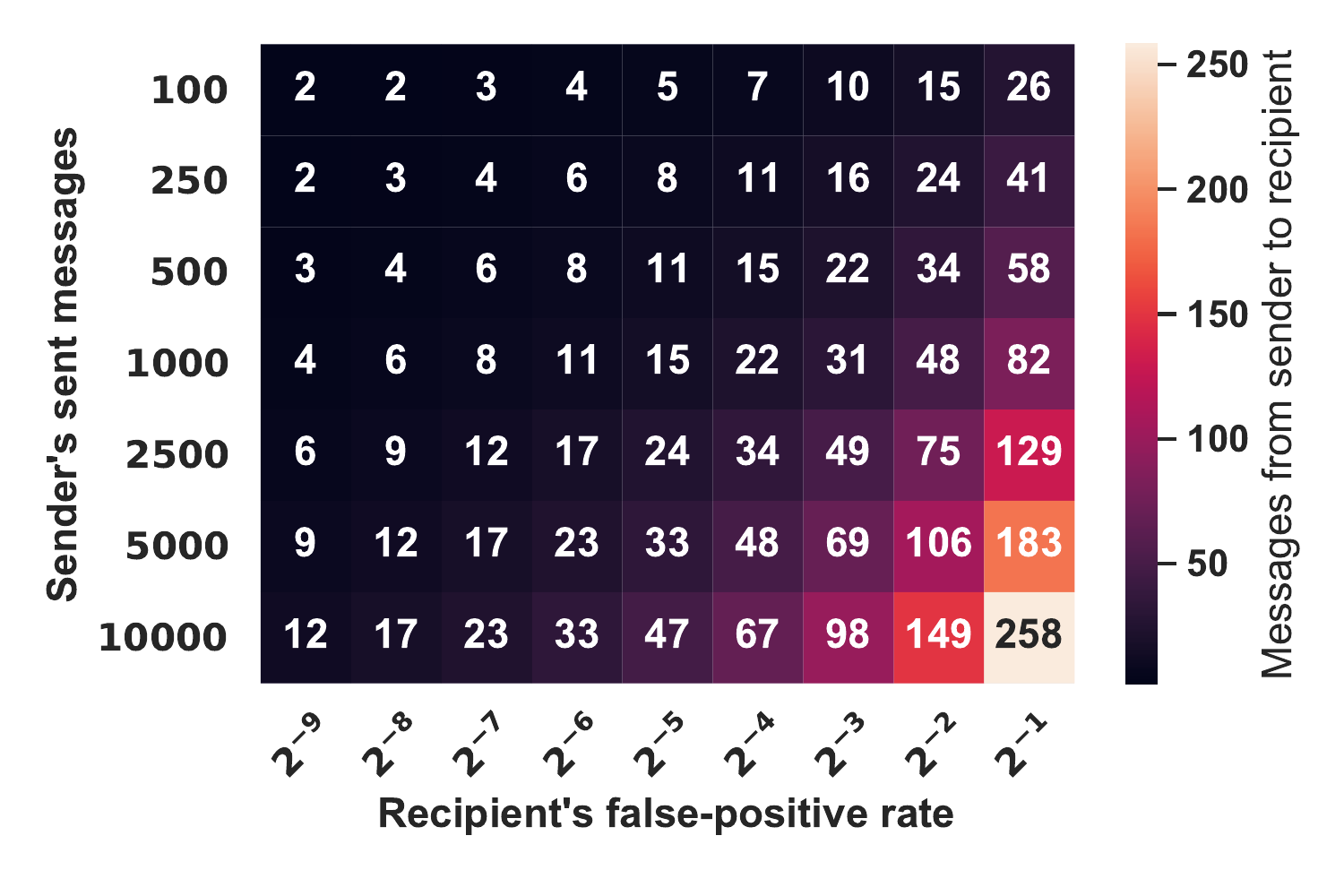}
  \caption{\centering Relationship anonymity for \linebreak
  $100\leq\mathsf{out}(s)$.}
  \label{fig:relanonymityA}
\end{subfigure}
\hfill
\begin{subfigure}{.49\textwidth}
  \centering
  \includegraphics[width=\linewidth]{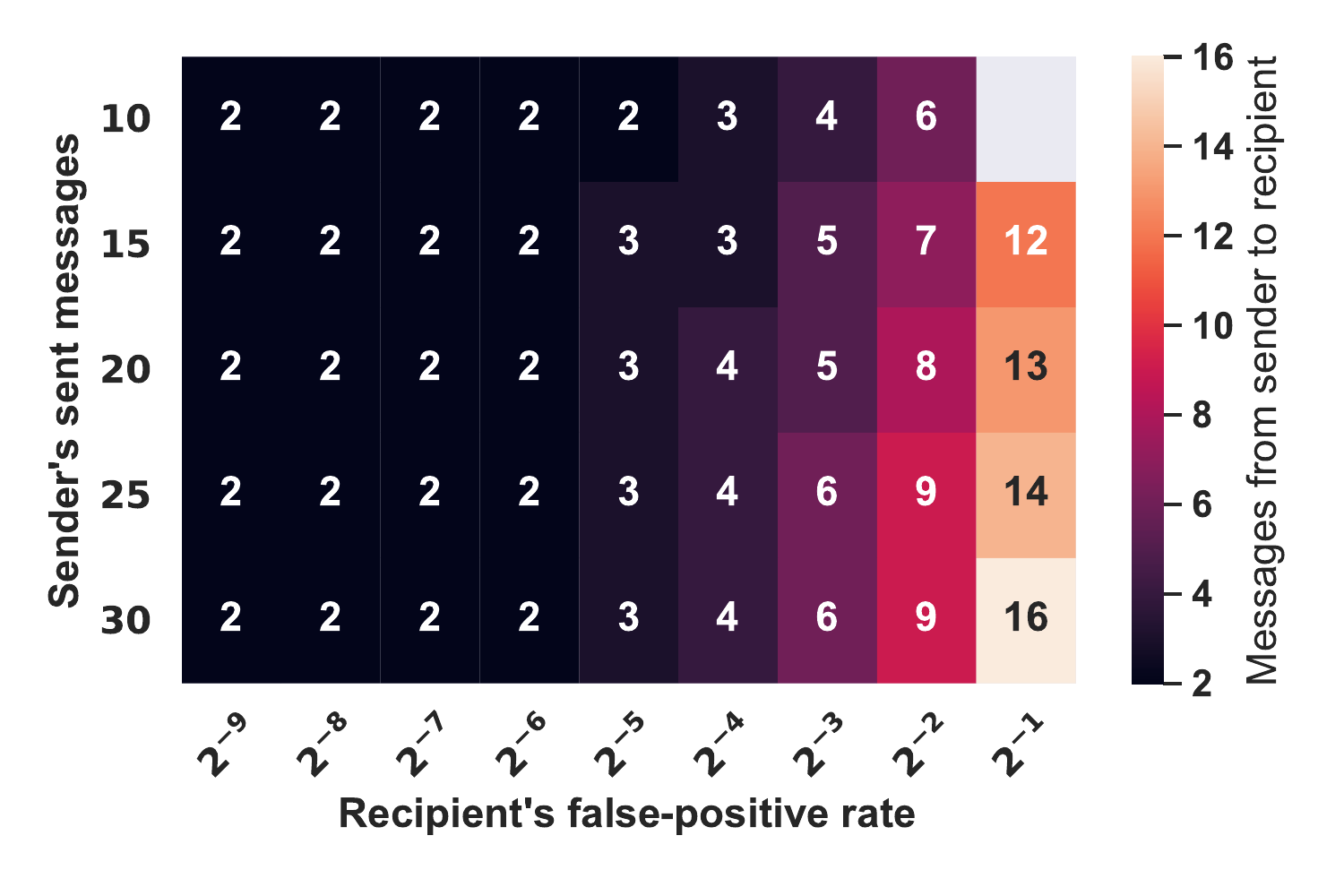}
  \caption{\centering Relationship anonymity for \linebreak $\mathsf{out}(s)\leq30$.}
  \label{fig:relanonymityB}
\end{subfigure}
\caption{The minimum number of messages between a pair of users that statistically reveal the relationship of the communicating users (significance level $1\%$). }
\label{fig:relanonymity2}
\end{figure}

\paragraph{Game Theoretic Analysis.}
Incentive compatibility has the utmost importance in decentralized and privacy-enhancing technologies. Therefore, we present a game-theoretic study of the FMD protocol concerning relationship anonymity. We believe applying game theory to FMD by itself is a fruitful and over-arching direction. Our goal, besides conducting a preliminary analysis, is to raise interest and fuel future research towards this direction. The formalization of the game as well as the corresponding theorems and proofs can be found in Appendix~\ref{sec:gametheory}.


\subsection{Temporal detection ambiguity}\label{sec:detambiguity}

The FMD scheme is required to satisfy the security notion of \emph{detection ambiguity} devised by Beck et al.~\cite{beckfuzzy}. Namely, for any message that yields a match for a detection key, the server should not be able to decide whether it is a true or a false-positive match. This definition is formalized for a single incoming message in isolation. Yet, the detector server can \emph{continuously} observe the stream of incoming messages.\footnote{As an illustrative example collected from a real communication system, see Figure~\ref{fig:detectionambiguityinanepoch}. }
Consequently, the server might be able to assess whether the user has received a message in \emph{a certain time interval}. 
To capture this time-dependent aspect, we relax detection ambiguity and coin the term \emph{temporal detection ambiguity}. Informally, no adversary should be able to tell in a given time interval having access to all incoming flag ciphertexts whether a user received an incoming true-positive match. We provide the formal definition in Appendix~\ref{sec:APPPRI}, and we empirically study temporal detection ambiguity on real communication data in Section~\ref{sec:experiments}. In Section~\ref{sec:dp}, we measure the level of privacy protection the number of incoming messages enjoys from a differential privacy angle.

Any message that enters the communication system yields a match to a detection key according to its set false-positive rate. Specifically, the number of false-positive matches acquired by user $u$'s detection key follows a $\mathsf{Binom}(M-\mathsf{in}(u),p(u))$ distribution. Similarly to Section~\ref{sec:relanon}, if $M$ is large, then we can approximate the number of false-positive matches with a $\mathcal{N}(p(u) M,p(u) (1-p(u)) M)$ distribution and use statistical tests to assess that the number of downloaded messages by a recipient is statistically far from the expected number of downloaded messages. More precisely, the adversary can statistically test whether $tag(u)$ could have been drawn from $\mathcal{N}(p(u) M,p(u)(1-p(u))M)$ (the approximation of $\mathsf{Binom}(M,p(u))$). We observe that in an epoch, a user should have either large false-positive rates or a small number of incoming messages to provide temporal detection ambiguity, shown in Figure~\ref{fig:detectionambiguity}.
\section{Differential Privacy Analysis}\label{sec:dp}

Differential privacy (DP)~\cite{dwork2006differential} is a procedure for sharing information about a dataset by publishing statistics of it while withholding information about single data points. DP is formalized in Definition~\ref{def:DP}; the core idea is to ensure that an arbitrary change on any data point in the database has a negligible effect on the query result. Hence, it is infeasible to infer much about any data point.

\begin{definition}[Differential Privacy~\cite{dwork2006differential}]
    \label{def:DP}
    An algorithm $A$ satisfies $\varepsilon$-differential privacy if for all $S\subseteq Range(A)$ and every input pair $D$ and $D^\prime$ differing in a single element Equation \ref{eq:DP} holds.
    \begin{equation}
        \label{eq:DP}
        \Pr(A(D)\in S)\le e^\varepsilon\cdot\Pr(A(D^\prime)\in S).
    \end{equation}
\end{definition}

\paragraph{Personalized Existing Edge DP. }
A widely used relaxation of the above definition is $(\varepsilon,\delta)$-DP, where Equation~\ref{eq:DP} is extended with a small additive term $\delta$ at the right end. 
There are over 200 modifications of DP~\cite{desfontaines2020sok}, we combined several to make it suitable for FMD. Concretely, we create a novel definition called \emph{Personalized Existing Edge DP (PEEDP)} (formally defined in Definition~\ref{def:peedp})\footnote{We elaborate more on various DP notions in Appendix~\ref{sec:DPAPP}.} by combining four existing notions. We utilize \emph{edge-DP}~\cite{hay2009accurate} which applies DP to communication graphs: $D$ and $D^\prime$ are the original communication graphs with and without a particular edge respectively, and $S$ is a set of graphs with fuzzy edges included. Furthermore, we apply \emph{personalized DP}~\cite{jorgensen2015conservative}, which allocates different level of protection to incoming messages, as in FMD the users' false positive rates could differ. 

Hiding the presence or absence of a message is only possible by explicitly removing real messages and adding fuzzy ones to the communication graph, which is indistinguishable from real ones. This setting (i.e., protecting existence and not value) corresponds to \emph{unbounded DP}~\cite{kifer2011no}.
Hence, as also noted in~\cite{beckfuzzy}, without a false negative rate (which would directly contradict correctness), FMD cannot satisfy DP: fuzzy messages can only hide the presence of a message not the absence. To tackle this imbalance, we utilize \emph{asymmetric DP}~\cite{takagi2021asymmetric} which only protects some of the records determined by policy $P$. It only differs from Definition \ref{def:DP} in the relationship of $D$ and $D^\prime$ as Equation~\ref{eq:DP} should only hold for \textit{every input pair $D$ and $D^\prime$ where later is created by removing in $D$ a single sensitive record defined by $P$}. By combining all these DP notions, we can formulate our PEEDP definition. 

\begin{definition}[$\overline{\varepsilon}$- Personalized Existing Edge Differential Privacy]
    \label{def:peedp}
    An algorithm $A$ satisfies $\overline{\varepsilon}$-PEEDP (where $\overline{\varepsilon}$ is an element-wise positive vector which length is equal with the amount of nodes in $D$) if Equation \ref{eq:DP_new} holds for all $S\subseteq Range(A)$ and every input graphs $D$ and $D^\prime$ where later is created by removing in $D$ a single incoming edge of user $u$.
    \begin{equation}
        \label{eq:DP_new}
        \Pr(A(D)\in S)\le e^{\overline{\varepsilon}_u}\cdot\Pr(A(D^\prime)\in S).
    \end{equation}
\end{definition}

Once we formalized a suitable DP definition for FMD, it is easy to calculate the trade-off between efficiency (approximated by $p(u)$) and privacy protection (measured by $\varepsilon_u$). This is captured in Theorem~\ref{th:osdp} (proof can be found in Appendix~\ref{sec:DPAPP}).

\begin{theorem}
    \label{th:osdp}
    If we assume the distribution of the messages are IID then FMD satisfy $\left[\log\frac1{p(u)}\right]_{u=1}^U$-PEEDP. 
\end{theorem}

Therefore, detection rates $p(u)=\{0.5^0,0.5^1,0.5^2,0.5^4,0.5^8\}$ in FMD correspond to $\varepsilon_u=\{0.000,0.693,1.386,2.773,5.545\}$ in $\overline{\varepsilon}$-PEEDP. 
Clearly, perfect protection (i.e., $\varepsilon_u=0$) is reached only when all messages are downloaded (i.e., $p(u)=1$). On the other hand, the other $\varepsilon$ values are much harder to grasp: generally speaking, privacy-parameter below one is considered strong with the classic DP definition. As PEEDP only provides a relaxed guarantee, we can postulate that the privacy protection what FMD offers is weak.

\paragraph{Protecting the Number of Incoming Messages. }
In most applications, e.g., anonymous messaging or stealth payments, we want to protect the number of incoming messages of the users, $in(u)$. 
Intuitively, the server observes $tag(u)\sim in(u)+\mathsf{Binom}(M-in(u),p(u))$ where (with sufficiently large $M$) the second term can be thought of as Gaussian-noise added to mask $in(u)$, a common technique to achieve $(\varepsilon,\delta)$-DP. 
Consequently, FMD does provide $(\varepsilon_u,\delta_u)$-DP\footnote{Note that this is also a personalized guarantee as in~\cite{jorgensen2015conservative}.} for the number of incoming messages of user $u$, see Theorem~\ref{th:DP} (proof in Appendix~\ref{sec:DPAPP}).

\begin{theorem}
    \label{th:DP}
    If we assume the distribution of the messages is IID than the FMD protocol provides $(\varepsilon_u,\delta_u)$-DP for the number of incoming messages $in(u)$ of user $u$ where $\delta_u=\max_u(p(u), 1-p(u))^{M-in(u)}$ and
    \begin{equation*}
        \varepsilon_u=\log\left[\max_u\left(\frac{p(u)\cdot(M-2\cdot in(u))}{(1-p(u))\cdot(in(u)+1)},\frac{(1-p(u))\cdot(M-in(u))}{p(u)}\right)\right].
    \end{equation*}
\end{theorem}

\begin{table}
    \centering
    \begin{tabular}{c||c|ccc||c|ccc}
        $M$ & 100 & 100 & 100 & \textbf{200} & \num{1000000} & \num{1000000} & \num{1000000} & $\mathbf{2}\hspace{0.05cm}\mathbf{000}\hspace{0.05cm}\mathbf{000}$ \\
        $in(u)$ & 10 & 10 & \textbf{20} & 10 & 100 & 100 & \textbf{1000} & 100 \\
        $p(u)$ & $0.5^4$ & $\mathbf{0.5^2}$ & $0.5^4$ & $0.5^4$ & $0.5^8$ & $\mathbf{0.5^4}$ & $0.5^8$ & $0.5^8$ \\
        \midrule
        $\varepsilon_u$ & 7.2 & 5.6 & 7.1 & 8.0 & 19.4 & 16.5 & 19.4 & 20 \\
        $\delta_u$ & 3e-3 & 6e-12 & 6e-3 & 5e-6 & 1e-1700 & 1e-28027 & 1e-1699 & 1e-3400 \\
    \end{tabular}
    \vspace{0.1cm}
    \caption{Exemplary settings to illustrate the trade-off between the false-positive rate $p(u)$ and the privacy parameters of $(\varepsilon,\delta)$-differential privacy for protecting the number of incoming messages.}
    \label{tab:DP_ex}
\end{table}

To illustrate our results, we provide some exemplary settings in Table \ref{tab:DP_ex} and show how the false positive rate $p(u)$ translates into $\varepsilon_u$ and $\delta_u$. It is visible that increasing the detection rate does increase the privacy protection (i.e., lower $\varepsilon_u$ and $\delta_u$), and increasing the overall and incoming messages result in weaker privacy parameter $\varepsilon_u$ and $\delta_u$ respectively. These results suggest, that even the number of incoming messages does not enjoy sufficient (differential) privacy protection in FMD, as the obtained values for $\varepsilon_u$ are generally considered weak.

\section{Evaluation} \label{sec:experiments}

We evaluate the relationship anonymity and temporal detection ambiguity guarantees of FMD through simulations on data from real communication systems.\footnote{
The simulator can be found at \url{https://github.com/seresistvanandras/FMD-analysis}}
We chose two real-world communication networks that could benefit from implementing and deploying FMD on top of them.

\begin{itemize}
    \item \textbf{College Instant Messaging (IM)~\cite{panzarasa2009patterns}.} 
    This dataset contains the instant messaging network of college students from the University of California, Irvine. 
    The graph consists of $\num{1899}$ nodes and $\num{59835}$ edges that cover $193$ days of communication.
    \item \textbf{EU Email~\cite{paranjape2017motifs}.} 
    This dataset is a collection of emails between members of a European research institution. 
    The graph has $\num{986}$ nodes and $\num{332334}$ edges. It contains the communication patterns of $803$ days.
\end{itemize}

Users are roughly distributed equally among major Information Privacy Awareness categories~\cite{soumelidou2021towards}, thus for each node in the datasets, we independently and uniformly at random chose a false-positive rate from the set $\{2^{-l}\}_{l=1}^7$. Note that the most efficient FMD scheme only supports false-positive rates of the form $2^{-l}$. Moreover, for each message and user in the system, we added new "fuzzy" edges to the graph according to the false-positive rates of the messages' recipients. The server is solely capable of observing the message-user graph with the added "fuzzy" edges that serve as cover traffic to enhance the privacy and anonymity of the users. We run our experiments $10$-fold where on average, there are around $16$ and $48$ million fuzzy edges for the two datasets, i.e., a randomly picked edge (the baseline) represents a genuine message with $\ll 1\%$.

\begin{figure}
\centering
\begin{subfigure}{.49\textwidth}
  \centering
  \includegraphics[width=\linewidth]{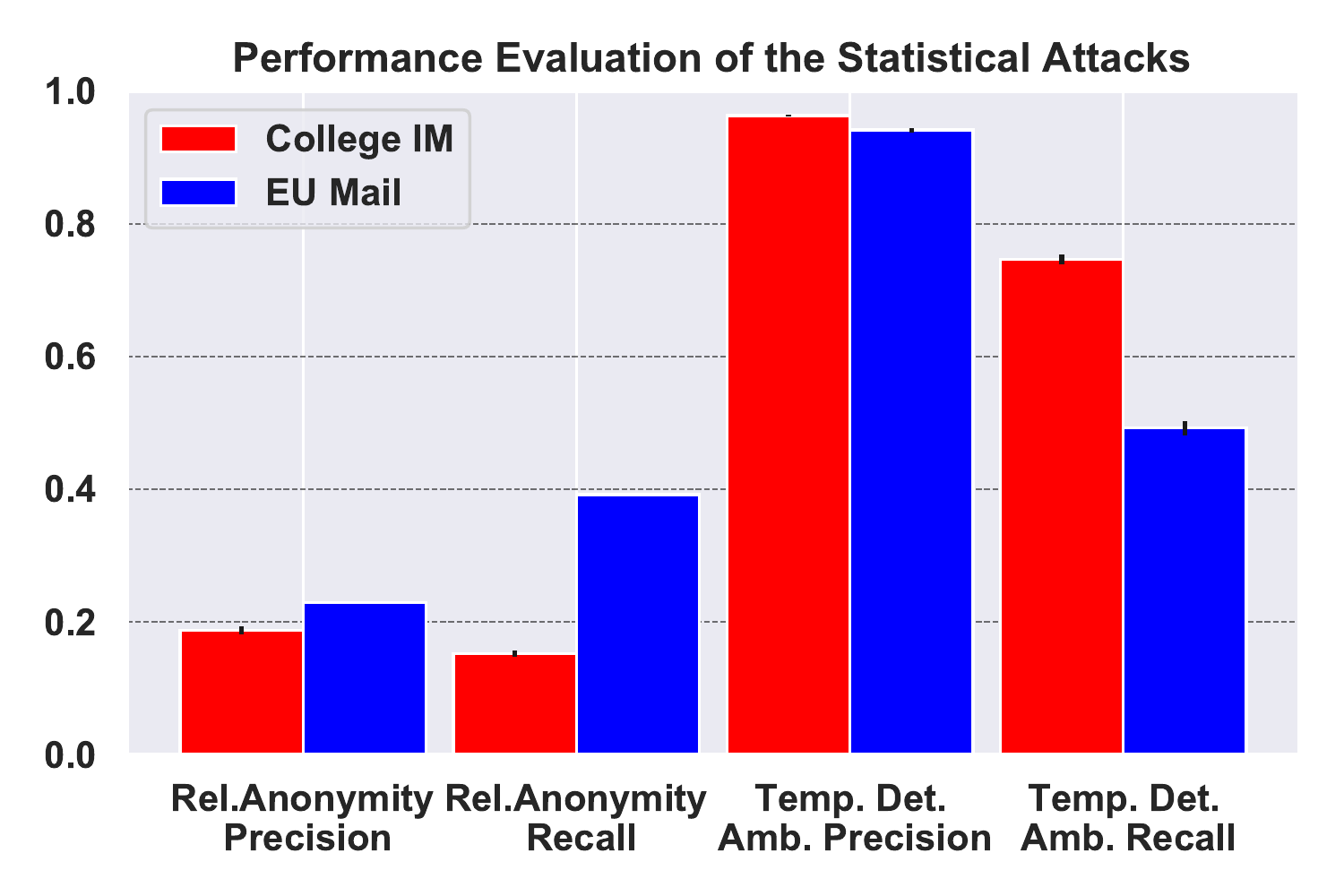}
  \caption{The precision and recall of the statistical tests breaking relationship anonymity and temporal detection ambiguity, cf. Section~\ref{sec:relanon} and~\ref{sec:detambiguity}. }
  \label{fig:evaluation}
\end{subfigure}
\hfill
\begin{subfigure}{.49\textwidth}
  \centering
  \includegraphics[width=\linewidth]{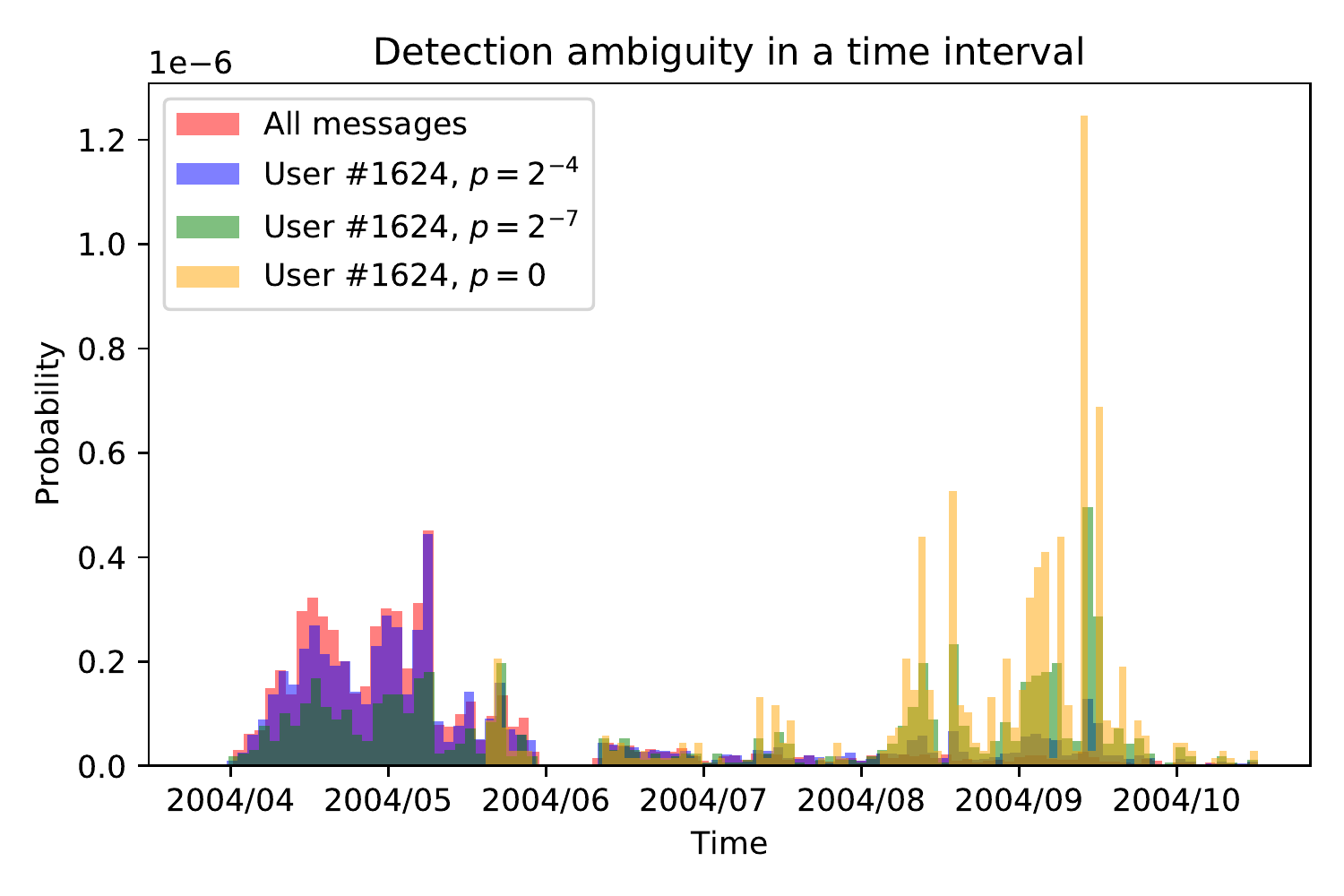}
  \caption{Temporal probability distribution of receiving a fuzzy tag for various false-positive detection rates. The exemplary user is taken from the College IM dataset.}
  \label{fig:detectionambiguityinanepoch}
\end{subfigure}
\caption{Privacy guarantees of FMD in simulations on real communication systems.}
\label{fig:evaluationMain}
\end{figure}

\subsection{Uncovering the relationship graph} \label{sec:exprelationanonymity}

The server's goal is to uncover the original social graph of its users, i.e., to expose the communicating partners. The relationship anonymity of a sender and a receiver can be easily uncovered by the statistical test introduced in Section~\ref{sec:relanon} especially if a user is receiving multiple messages from the same sender while having a low false-positive rate.
We found that statistical tests produce a $0.181$ and $0.229$ precision with $0.145$ and $0.391$ recall on average in predicting the communication links between all the pairs of nodes in the College IM and EU Email datasets, respectively, see Figure~\ref{fig:evaluation}. 
The results corresponding to the EU Email dataset are higher due to the increased density of the original graph. These results are substantial as they show the weak anonymization power of FMD in terms of relationship anonymity.

\begin{figure}
\centering
\begin{subfigure}{.5\textwidth}
  \centering
  \includegraphics[width=\linewidth]{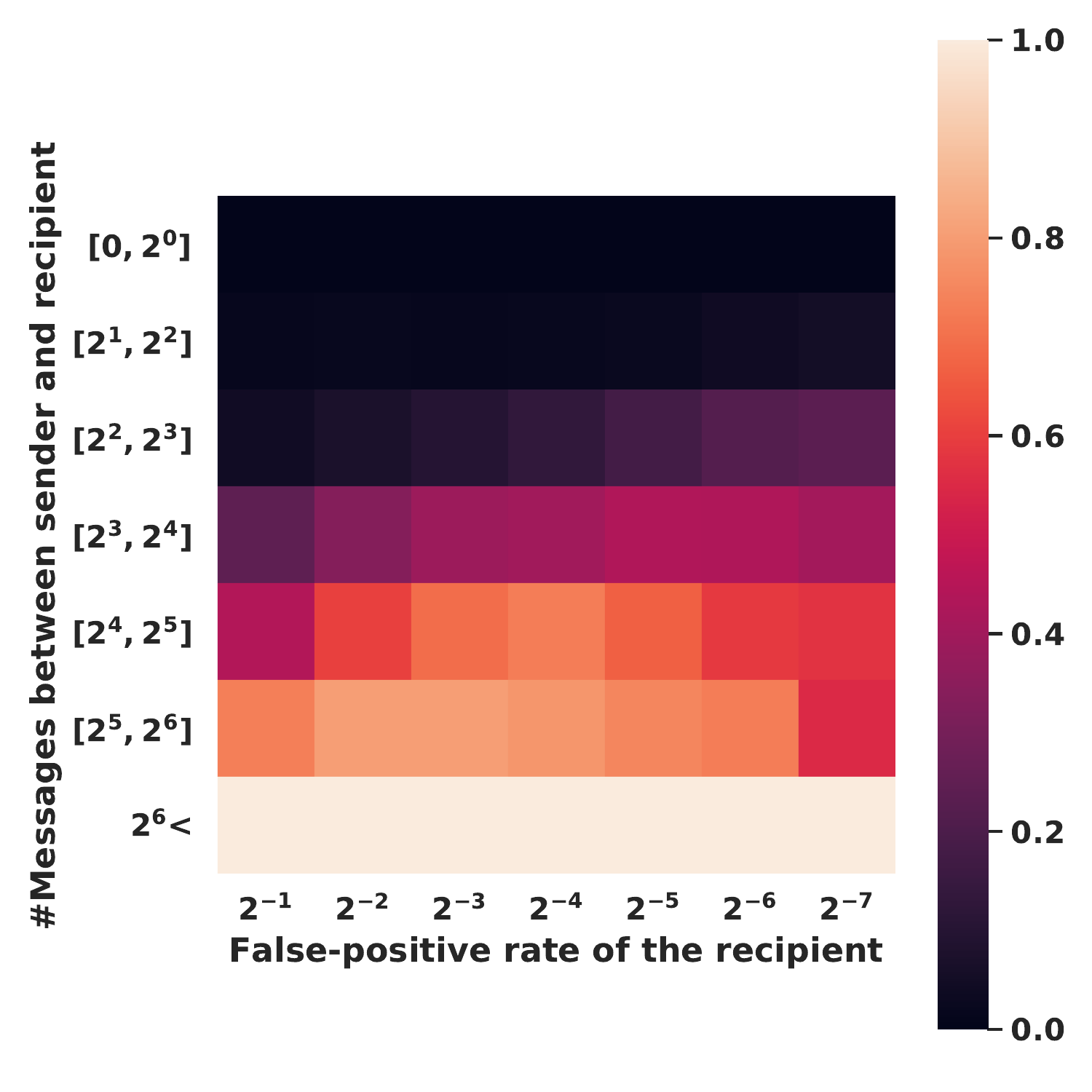}
\end{subfigure}%
\begin{subfigure}{.5\textwidth}
  \centering
  \includegraphics[width=\linewidth]{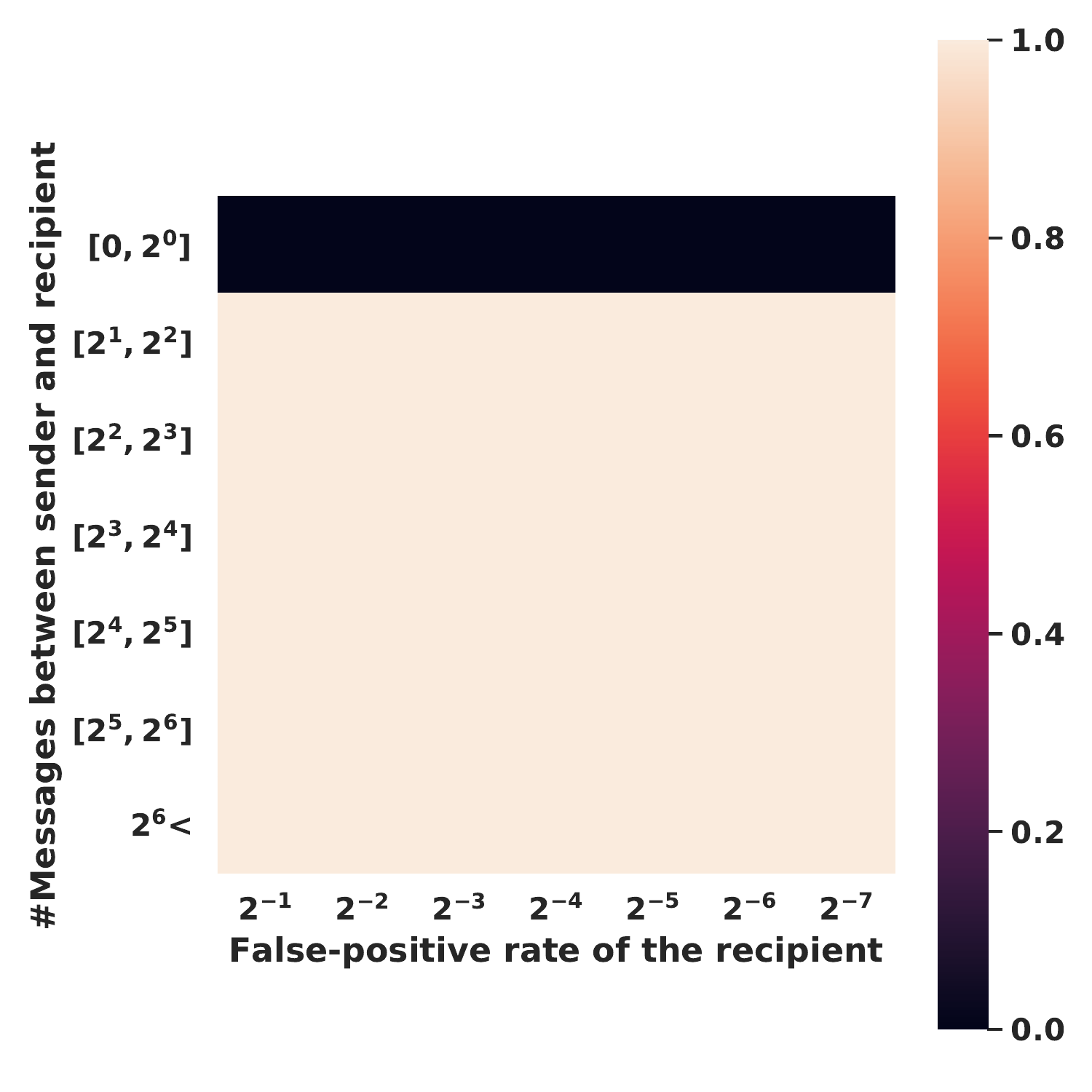}
\end{subfigure}
\caption{Recall (left) and precision (right) of the statistical tests in breaking relationship anonymity (see Section~\ref{sec:relanon}) in simulations on the College IM dataset.}
\label{fig:precisionGranularEvalPlus}
\end{figure}

Specifically, communication relationships where merely a single message has been exchanged remain undetected by the applied statistical tests, cf. Figure~\ref{fig:precisionGranularEvalPlus}. However, note that for every other pairs of users, neither of the analyzed datasets yields false positives by the used statistical tests. These simulation results demonstrate that relationship anonymity is effectively maintained against statistical attacks if each user sends only a single message from the server's point of view. This can be achieved by cryptographic tools or anonymous communication systems, e.g., Tor. On the other hand, recurrent communication relationships reveal the relationship of communicating peers. Thus, relationship anonymity is breached with perfect precision, cf. Figure~\ref{fig:precisionGranularEvalPlus}. Simulation results confirm our intuition as well. Namely, statistical tests produce higher recall for nodes with lower false-positive detection rates, while they are less effective for communicating pairs that exchanged very few messages.

\subsection{Breaking Temporal Detection Ambiguity}\label{sec:expdetectionambiguity}

We empirically quantify whether users can deny that they received an incoming (true positive) message. We consider $\num{25000}$ randomly selected messages with the corresponding fuzzy edges as one epoch. The server tried to assess using statistical tests (see Section~\ref{sec:detambiguity}) that a user has received an incoming message. 
The intuition is that users receive messages heterogeneously concerning time. Hence, surges in incoming traffic might not be adequately covered by fuzzy edges for users with low false-positive rates. Thus, these messages could be tight to the receiver with high probability, see Figure~\ref{fig:detectionambiguityinanepoch} for an illustrative example. 
Indeed, Figure~\ref{fig:evaluation} suggests that, in general, deniability can be broken effectively with high precision and recall. 
On the other hand, Figure~\ref{fig:detectionambiguityinanepoch} also shows that higher false-positive rates could provide enough cover traffic for messages within these conspicuous epochs, which is in line with the findings presented in Figure~\ref{fig:detectionambiguity}.
\section{Conclusion}\label{sec:conclusion}

In this paper, we present a privacy and anonymity analysis 
of the recently introduced Fuzzy Message Detection scheme. 
Our analysis is thorough as it covers over three directions. 
Foremost, an information-theoretical analysis was carried out concerning recipient unlinkability, relationship anonymity, and temporal detection ambiguity. 
It is followed by a differential privacy analysis which leads to a novel privacy definition.
Finally, we gave an exhaustive simulation based on real-world data. 
Our findings facilitate proper parameter selection and the deployment of the FMD scheme into various applications. 
Yet, we also raise concerns about the guarantees what FMD provides and questions whether it is adequate/applicable for many real-world scenarios. 

\paragraph{Limitations and Future Work. }
Although far-reaching, our analysis only scratches the surface of what can be analyzed concerning FMD, and substantial work and important questions remain as future work. 
Thus, a hidden goal of this paper is to fuel discussions about FMD so it can be deployed adequately for diverse scenarios. 
Concretely, we formulated a game only for one privacy property and did not study the Price of Stability/Anarchy. 
Concerning differential privacy, our assumption about the IID nature of the edges in a communication graph is non-realistic. At the same time, the time-dependent aspect of the messages is not incorporated in our analysis via Pan-Privacy.

\paragraph{Acknowledgements.}
We thank our shepherd Fan Zhang and our anonymous reviewers for helpful comments in preparing the final
version of this paper.
We are grateful to Sarah Jamie Lewis for inspiration and publishing the data sets. We thank Henry de Valence and Gabrielle Beck for fruitful discussions.
Project no. 138903 has been implemented with the support provided by the Ministry of Innovation and Technology from the NRDI Fund, financed under the FK\_21 funding scheme.
The research reported in this paper and carried out at the BME has been supported by the NRDI Fund based on the charter of bolster issued by the NRDI Office under the auspices of the Ministry for Innovation and Technology.

\bibliography{sample}
\bibliographystyle{plain}

\appendix
\section{FMD in more details}
\label{sec:app_fmd}

The fuzzy message detection scheme consists of the following five probabilistic polynomial-time algorithms $(\mathsf{Setup},\mathsf{KeyGen},\mathsf{Flag},\mathsf{Extract},\mathsf{Test})$. In the following, let $\mathcal{P}$ denote the set of attainable false positive rates.

\begin{description}
\item[$\mathsf{Setup}(1^{\lambda})$]$\xrightarrow{\$}\mathsf{pp}$. Global parameters $\mathsf{pp}$ of the FMD scheme are generated, i.e., the description of a shared cyclic group.
\item[$\mathsf{KeyGen}_{\mathsf{pp}}(1^{\lambda})$]$\xrightarrow{\$}(pk,sk)$. This algorithm is given the global public parameters and the security parameter and outputs a public and secret key. 
\item[$\mathsf{Flag}(pk)$]$\xrightarrow{\$}C.$ This randomized algorithm given a public key $pk$ outputs a flag ciphertext $C$.
\item[$\mathsf{Extract}(sk,p)$]$\xrightarrow[]{}dsk$. Given a secret key $sk$ and a false positive rate $p$ the algorithm extracts a detection secret key $dsk$ iff. $p\in\mathcal{P}$ or outputs $\bot$ otherwise.
\item[$\mathsf{Test}(dsk,C)$]$\xrightarrow[]{}\{0,1\}$. The test algorithm given a detection secret key $dsk$ and a flag ciphertext $C$ outputs a detection result.
\end{description}

An FMD scheme needs to satisfy three main security and privacy notions: correctness, fuzziness and detection ambiguity. For the formal definitions of these, we refer to~\cite{beckfuzzy}.
The toy example presented in Figure~\ref{fig:fmd_visual} is meant to illustrate the interdependent nature of the privacy guarantees achieved by the FMD scheme.

\begin{figure}
    \centering
    \includegraphics[width=0.7\linewidth,trim={0.7cm 5.2cm 0cm 2cm},clip,scale=0.5]{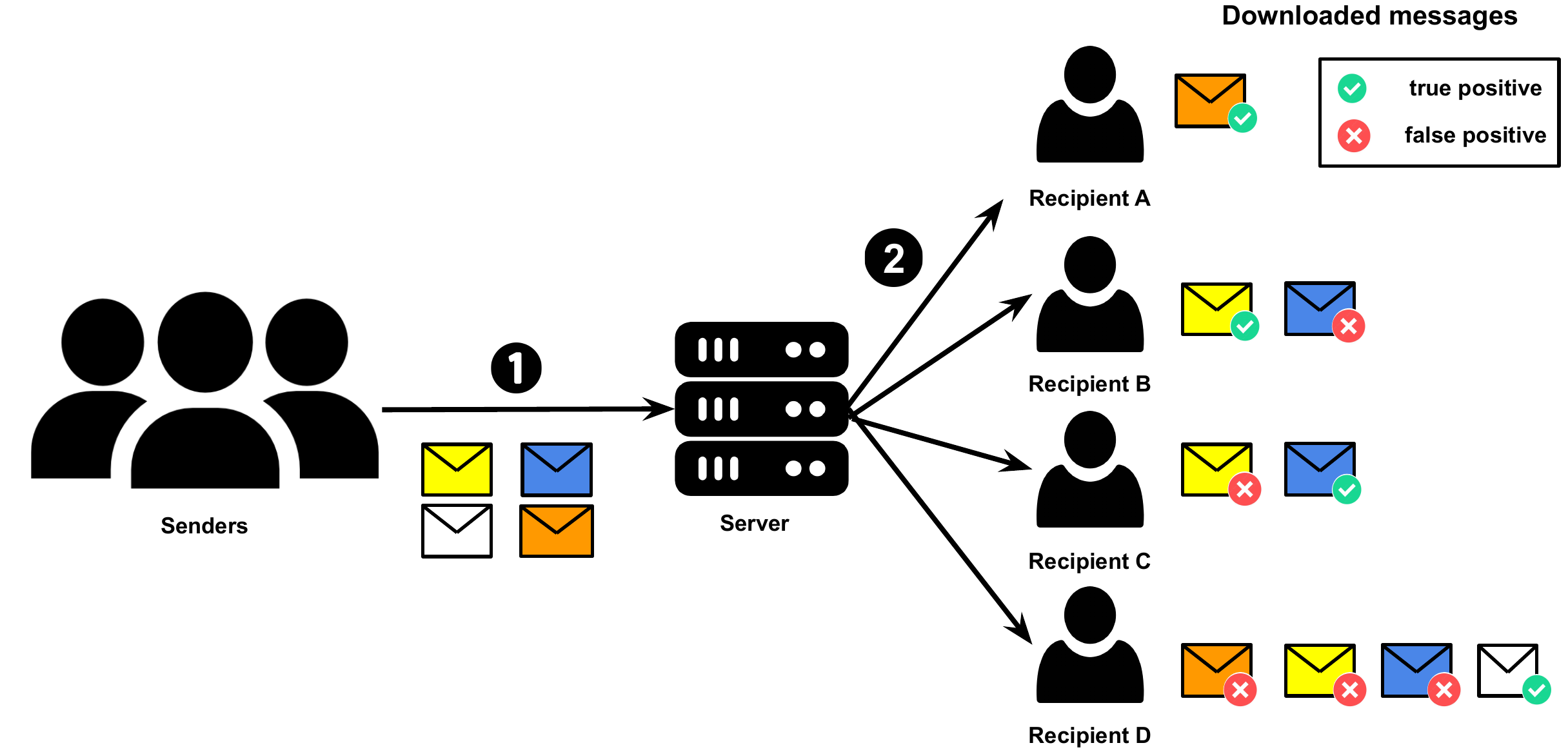}
    \caption{A toy example of the FMD scheme. \encircle{1} Several senders post anonymous messages to the untrusted server. \encircle{2} Whenever recipients come online, they download messages that correspond to them (some false positve, some true positive). Recipient A,B,C and D have a false positive rate $0,\frac{1}{3},\frac{1}{3},1$, respectively. Note that the server can map the messages that belong to A and D. However, the messages of Recipient B and C are $2$-anonymous.}
    \label{fig:fmd_visual}
\end{figure}
\section{Formal definitions of security and privacy guarantees}
\label{sec:APPPRI}



\begin{figure}
    \centering
 \begin{tcolorbox}[
    standard jigsaw,
    opacityback=0]
{\bf The recipient unlinkability $\mathcal{G}^{RU}_{\mathcal{A},\Pi}(\lambda)$ game}

{
\begin{enumerate}
    \item Adversary $\mathcal{A}$ selects target recipients $u_0, u_1$ and a target sender $u_2$.
    \item Challenger $\mathcal{C}$ instructs sender $u_2$ to send a message to $u_c$ for $c\xleftarrow{\$}\{0,1\}$.
    \item $\mathcal{C}$ uniformly at random generates a challenge bit $b\xleftarrow{\$}\{0,1\}$. If $b=0$, $\mathcal{C}$ instructs $u_2$ to send a message to $u_c$. Otherwise, instructs $u_2$ to send a message to $u_{1-c}$.   
    \item $\mathcal{A}$ observes network traffic and flag ciphertexts and outputs $b'$.
    \item Output $1$, iff. $b=b'$, otherwise $0$.
\end{enumerate}
}
\end{tcolorbox}
    \caption{The security game for the anonymity notion of recipient unlinkability.}
    \label{fig:rugame}
\end{figure}



\begin{definition}[Temporal Detection Ambiguity]
An anonymous communication protocol $\Pi$ satisfies temporal detection ambiguity if for all probabilistic polynomial-time adversaries $\mathcal{A}$ there is a negligible function $\mathsf{negl}(\cdot)$ such that
\begin{equation}
    \Pr[\mathcal{G}^{TDA}_{\mathcal{A},\Pi}(\lambda)=1]\leq\frac{1}{2}+\mathsf{negl}(\lambda),
\end{equation}
where the temporal detection ambiguity game $\mathcal{G}^{TDA}_{\mathcal{A},\Pi}(\cdot)$ is defined below.
\end{definition}

\begin{figure}
    \centering
\begin{tcolorbox}[
    standard jigsaw,
    opacityback=0]
{\bf The temporal detection ambiguity $\mathcal{G}^{TDA}_{\mathcal{A},\Pi}(\lambda)$ game}

{
\begin{enumerate}
    \item Adversary $\mathcal{A}$ selects a target recipient $u_0$.
    \item Challenger $\mathcal{C}$ uniformly at random generates a challenge bit $b\xleftarrow{\$}\{0,1\}$. If $b=0$, $\mathcal{C}$ picks $k\xleftarrow{\$}[1,2,\dots,U]$ and instructs sender $u_k$ to send a message to $u_0$. Otherwise, the challenger does nothing. 
    \item The anonymous communication protocol $\Pi$ remains functional for a certain period of time, i.e., users keep sending messages using $\Pi$.
    \item $\mathcal{A}$ observes network traffic and flag ciphertexts and outputs $b'$.
    \item Output $1$, iff. $b=b'$, otherwise $0$.
\end{enumerate}
}
\end{tcolorbox}
    \caption{The security game for the privacy notion of temporal detection ambiguity}
    \label{fig:tdagame}
\end{figure}
\section{Differential Privacy Relaxations \& Proofs}
\label{sec:DPAPP}

Our novel DP notion called PEEDP (short for Personalized Existing Edge DP) is an instance of $d$-privacy~\cite{chatzikokolakis2013broadening}, which generalizes the neighbourhood of datasets (on which the DP inequality should hold) to an arbitrary metric $d$ defined over the input space. Yet, instead of a top-down approach where we are presenting a complex metric to fit our FMD use-case, we follow a bottom-up approach and show the various building blocks of our definition. PEEDP is a straight forward combination of unbounded DP~\cite{kifer2011no}, edge-DP~\cite{hay2009accurate}), assymetric DP~\cite{takagi2021asymmetric}, and personalized DP~\cite{jorgensen2015conservative}. Although Definition~\ref{def:peedp} is appropriate for FMD, it does not capture the FMD scenarios fully as neither time-dependent nature of the messages nor the dependencies and correlations between them are taken into account.

The first issue can be tackled by integrating other DP notions into PEEDP which provide guarantees under continuous observation (i.e., stream-data), such as pan-privacy~\cite{dwork2010pan}. Within this streaming context several definitions can be considered: user-level~\cite{dwork2010differential} (to protect the presence of users), event-level~\cite{dwork2010new} (to protect the presence of messages), and $w$-event level~\cite{kellaris2014differentially} (to protect the presence of messages within time windows). 

The second issue is also not considered in Theorem~\ref{th:osdp} as we assumed the messages are IID, while in a real-world applications this is not necessarily the case. Several DP notions consider distributions, without cherry-picking any we refer the readers to two corresponding surveys \cite{desfontaines2020sok,zhang2020correlated}. We leave it as a future work to tweak our definition further to fit into these contexts. 
\begin{proof}[of Theorem~\ref{th:osdp}]
	Due to the IID nature of the messages it is enough to show that Equation \ref{eq:DP_new} holds for an arbitrary communication graph $D$ with arbitrary message $m$ of an arbitrary user $u$. The two possible world the adversary should not be able to differentiate between $D=D^\prime/\{m\}$, i.e., whether the particular message exists or not. 
	Due to the asymmetric nature of Definition~\ref{def:peedp} (i.e., it only protects the existence) Equation \ref{eq:0} does not need to be satisfied. On the other hand, if the message exists than Equation \ref{eq:1} and \ref{eq:2} must be satisfied where $S_1=$\{message $m$ is downloaded by user $u$\} and $S_2=$\{message $m$ is not downloaded by user $u$\}.
	
	\begin{gather}
		\label{eq:0}\Pr(A(D^\prime)\in S)\le e^{\varepsilon_u} \cdot \Pr(A(D)\in S)\\
		\label{eq:1}\Pr(A(D)\in S_1)\le e^{\varepsilon_u} \cdot \Pr(A(D^\prime)\in S_1)\\
		\label{eq:2}\Pr(A(D)\in S_2)\le e^{\varepsilon_u} \cdot \Pr(A(D^\prime)\in S_2)
	\end{gather}
	
	If we reformulate the last two equations with the corresponding probabilities we get $1\le e^{\varepsilon_u}\cdot p(u)$ and $0\le e^{\varepsilon_u}\cdot(1-p(u))$ respectively. While the second holds trivially the first corresponds to the formula in Theorem \ref{th:osdp}.	
	\hfill\qed
\end{proof}

\begin{proof}[of Theorem~\ref{th:DP}]
The users' number of incoming messages are independent from each other hence we can focus on a single user $u$. The proof follows the idea from~\cite{korolova2009releasing}\footnote{We present the proof for singleton sets, but it can be extended by using the following formula: $\frac{A+C}{B+D}<\max(\frac{A}{B},\frac{C}{D})$.}: we satisfy Equation \ref{eq:DP} (with $+\delta$ at the end) when $A(D)=tag(u)\sim D+\mathsf{Binom}(M-in(u),p(u))$ for $D=in(u)$ and $D^\prime=in(u)\pm1$, i.e., we show that the following Equation holds.

\begin{equation*}
    \begin{split}
        \Pr(A(D)=tag(u)\in S|D=in(u),M,p(u))\le\\
        e^\varepsilon\cdot\Pr(A(D^\prime)=tag^\prime(u)\in S|D^\prime=in(u)\pm1,M^\prime=M\pm1,p(u))+\delta\\
        \Rightarrow\hspace{1cm}\Pr(in(u)+\mathsf{Binom}(M-in(u),p(u))\in S)\le\\
        e^\varepsilon\cdot\Pr(in(u)\pm1+\mathsf{Binom}(M\pm1-(in(u)\pm1),p(u))\in S)+\delta\\
    \end{split}
\end{equation*}

First, we focus on $\delta$ and provide a lower bound originating from the probability on the left when $\Pr(\cdot)\le e^\varepsilon\cdot0+\delta$. 
This corresponds to two cases as seen in the Equation below: when $D^\prime=in(u)+1$ with $S=\{in(u)\}$ and when $D^\prime=in(u)-1$ with $S=\{M\}$. The corresponding lower bounds (i.e., probabilities) correspond to the event when user $u$ does not download any fuzzy messages and when user $u$ does downloads all messages respectively. Hence, the maximum of these are indeed a lower bound for $\delta$. 

\begin{gather*}
    \Pr(A(in(u))=in(u))\le e^\varepsilon\cdot\Pr(A(in(u)+1)=in(u))+\delta\Rightarrow (1-p(u))^{M-in(u)}\le\delta\\
    \Pr(A(in(u))=M)\le e^\varepsilon\cdot\Pr(A(in(u)-1)=M)+\delta\hspace{0.2cm}\Rightarrow\hspace{0.2cm}p(u)^{M-in(u)}\le\delta
\end{gather*}

Now we turn towards $\varepsilon$ and show that $(\varepsilon,0)$-DP holds for all subset besides the two above, i.e., when $S=\{in(u)+y\}$ with $y=[1,\dots,M-in(u)-1]$. First, we reformulate Equation \ref{eq:DP} as seen below. 

\begin{equation*}
        \frac{\Pr(in(u)+\mathsf{Binom}(M-in(u),p(u))\in S)}{\Pr(in(u)\pm1+\mathsf{Binom}(M-in(u),p(u))\in S)}\le
        e^\varepsilon
\end{equation*}

Then, by replacing the binomial distributions with the corresponding probability formulas we get the following two equations for $D^\prime=in(u)+1$ and $D^\prime=in(u)-1$ respectively. 

\begin{gather*}
    \frac{\binom{M-in(u)}{y}\cdot p(u)^y\cdot(1-p(u))^{M-in(u)-y}}{\binom{M-in(u)}{y-1}\cdot p(u)^{y-1}\cdot(1-p(u))^{M-in(u)-y+1}}=\frac{M-in(u)-y+1}{y}\cdot\frac{p(u)}{1-p(u)}\le e^\varepsilon\\
    \frac{\binom{M-in(u)}{y}\cdot p(u)^y\cdot(1-p(u))^{M-in(u)-y}}{\binom{M-in(u)}{y+1}\cdot p(u)^{y+1}\cdot(1-p(u))^{M-in(u)-y-1}}=\frac{y+1}{M-in(u)-y}\cdot\frac{1-p(u)}{p(u)}\le e^\varepsilon\\
\end{gather*}

Consequently, the maximum of these is the lower bound for $e^\varepsilon$. The first formula's derivative is negative, so the function is monotone decreasing, meaning that its maximum is at $y=in(u)+1$. On the other hand, the second formula's derivative is positive so the function is monotone increasing, hence the maximum is reached at $y=M-in(u)-1$. By replacing $y$ with these values respectively one can verify that the corresponding maximum values are indeed what is shown in Theorem \ref{th:DP}. 
\hfill\qed
\end{proof}
\section{Game-Theoretical Analysis}
\label{sec:gametheory}

Here --- besides a short introduction of the utilized game theoretic concepts --- we present a rudimentary game-theoretic study of the FMD protocol focusing on relationship anonymity introduced in Section~\ref{sec:privacyguarantees}. 
First, we formalize a game and highlight some corresponding problems such as the interdependence of the user's privacy. 
Then, we unify the user's actions and show the designed game's only Nash Equilibrium, which is to set the false-positive detection rates to zero, rendering FMD idle amongst selfish users. 
Following this, we show that a higher utility could been reached with altruistic users and/or by centrally adjusting the false-positive detection rates. 
Finally, we show that our game (even with non-unified actions) is a potential game, which have several nice properties, such as efficient Nash Equilibrium computation.

\begin{itemize}
    \item \textbf{Tragedy of Commons}~\cite{hardin1968tragedy}: users act according to their own self-interest and, contrary to the common good of all users, cause depletion of the resource through their uncoordinated action. 
    \item \textbf{Nash Equilibrium}~\cite{nash1950equilibrium}: every player makes the best/optimal decision for itself as long as the others’ choices remain unchanged. 
    \item \textbf{Altruism}~\cite{simon1993altruism}: users act to promote the others' welfare, even at a risk or cost to ourselves.
    \item \textbf{Social Optimum}~\cite{harsanyi1988general}: the user's strategies which maximizes social welfare (i.e., the overall accumulated utilities).
    \item \textbf{Price of Stability/Anarchy}~\cite{nshelevich2008price,koutsoupias1999worst}: the ratio between utility values corresponding to the best/worst NE and the SO. It measures how the efficiency of a system degrades due to selfish behavior of its agents.
    \item \textbf{Best Response Mechanism}~\cite{nisan2011best}: from a random initial strategy the players iteratively improve their strategies
\end{itemize}

Almost every multi-party interaction can be modeled
as a game. In our case, these decision makers are the
users using the FMD service. We assume the users bear some costs $C_u$ for downloading any message from the server. For simplicity we define this uniformly: if $f$ is the cost of retrieving any message for any user than $C_u=f\cdot tag(u)$. Moreover, we replace the random variable $tag(u)\sim in(u)+\mathsf{Binom}(M-in(u),p(u))$ with its expected value, i.e., $C_u=f\cdot(in(u)+p(u)\cdot(M-in(u)))$. 

Besides, the user's payoff should depend on whether any of the privacy properties detailed in Section \ref{sec:privacyguarantees} are not satisfied. For instance, we assume the users suffer from a privacy breach if relationship anonymity is not ensured, i.e., they uniformly lose $L$ when the recipient $u$ can be linked to any sender via any message between them.
In the rest of the section we slightly abuse the notation $u$ as in contrast to the rest of the paper we refer to the users as $u\in\{1,\dots,U\}$ instead of $\{u_0,u_1,\dots\}$.
The probability of a linkage via a particular message for user $u$ is $\alpha_u=\prod_{v\in\{1,\dots,U\}/u}(1-p(v))$. The probability of a linkage from any incoming message of $u$ is $1-(1-\alpha_u)^{in(u)}$.\footnote{It is only an optimistic baseline as it merely captures the trivial event when no-one downloads the a message from any sender $v$ besides the intended recipient $u$.} Based on these we define the FMD-RA Game.

\begin{definition}
The FMD-RA Game is a tuple $\langle\mathcal{N},\Sigma,\mathcal{U}\rangle$, where the set of players is $\mathcal{N}=\{1,\dots,U\}$, their actions are $\Sigma=\{p(1),\dots,p(U)\}$ where $p(u)\in[0,1]$ while their utility functions are $\mathcal{U}=\{\varphi_u(p(1),\dots,p(U))\}_{u=1}^U$ such that for $1\le u\le U$:
\begin{equation}
    \label{eq:util}
    \begin{split}
    \varphi_u=-L\cdot\left(1-\left(1-\alpha_u\right)^{in(u)}\right)-f\cdot(in(u)+p(u)\cdot(M-in(u))).
    \end{split}
\end{equation}
\end{definition}

It is visible in the utility function that the bandwidth-related cost (second term) depends only on user $u$'s action while the privacy-related costs (first term) depend only on the other user's actions. This reflects well that relationship anonymity is an interdependent privacy property~\cite{biczok2013interdependent} within FMD: by downloading fuzzy tags, the users provide privacy to others rather than to themselves. As a consequence of this tragedy-of-commons~\cite{hardin1968tragedy} situation, a trivial no-protection Nash Equilibrium (NE) emerges. Moreover, Theorem~\ref{th:trivialNE} also states this NE is unique, i.e., no other NE exists.

\begin{theorem}
\label{th:trivialNE}
    Applying no privacy protection in the FMD-RA Game is the only NE: $(p^*(1),\dots,p^*(U))=(0,\dots,0)$.
\end{theorem}

\begin{proof}
    First we prove that no-protection is a NE. If all user $u$ set $p(u)=0$ than a single user by deviates from this strategy would increased its cost. Hence no rational user would deviate from this point. In details, in Equation~\ref{eq:util} the privacy related costs is constant $-L$ independently from user $u$'s false-positive rate while the download related cost would trivially increase as the derivative of this function (shown in Equation~\ref{eq:deriv}) is negative.
    \begin{equation}
        \label{eq:deriv}
        \frac{\partial\varphi_u}{\partial p(u)}=-f\cdot(M-in(u))<0
    \end{equation}
    Consequently, $p^*=(p^*(1),\dots,p^*(U))=(0,\dots,0)$ is indeed a NE. Now we give an indirect reasoning why there cannot be any other NEs. Lets assume $\hat{p}=(\hat{p}(1),\dots,\hat{p}(U))$ is a NE. At this state any player could decrease its cost by reducing its false positive-rate which only lower the download related cost. Hence, $\hat{p}$ is not an equilibrium. 
    \hfill\qed
\end{proof}

This negative result highlights that in our simplistic model, no rational (selfish) user would use FMD; it is only viable when altruism~\cite{simon1993altruism} is present. On the other hand, (if some condition holds) in the Social Optimum (SO)~\cite{harsanyi1988general}, the users do utilize privacy protection. This means a higher total payoff could be achieved (i.e., greater social welfare) if the users cooperate or when the false-positive rates are controlled by a central planner. 
Indeed, according to Theorem~\ref{th:so} the SO$\not=$NE if, for all users, the cost of the fuzzy message downloads is smaller than the cost of the privacy loss. The exact improvement of the SO over the NE could be captured by the Price of Stability/Anarchy~\cite{nshelevich2008price,koutsoupias1999worst}, but we leave this as future work. 

\begin{theorem}
    \label{th:so}
    The SO of the FMD-RA Game is not the trivial NE and corresponds to higher overall utilities if $f\cdot(M-\max_u(in(u)))<L$.
\end{theorem}

\begin{proof}
    We show that the condition in the theorem is sufficient to ensure that SO$\not=$NE by showing that greater utility could be achieved with $0<p^\prime(u)$ than with $p(u)=0$. To do this we simplify out scenario and set $p(u)=p$ for all users. The corresponding utility function is presented in Equation \ref{eq:util2} while in Equation~\ref{eq:01} we show the exact utilities when $p$ is either 0 or 1. 
    
    \begin{gather}
        \label{eq:util2}
        \varphi_u(p)=
        - L \cdot (1-(1-(1-p)^{U-1})^{in(u)})
        - f \cdot (in(u) + p \cdot (M - in(u)))\\
        \label{eq:01}
        \varphi_u(0)=-L-f\cdot in(u)
        \hspace{1cm}
        \varphi_u(1)=-f\cdot M
    \end{gather}

    One can check with some basic level of mathematical analysis that the derivative of Equation~\ref{eq:util2} is negative at both edge of $[0,1]$ as $\frac{\partial\varphi_u(p)}{\partial p}(0)=\frac{\partial\varphi_u(p)}{\partial p}(1)=-f\cdot(M-in(u))$. This implies that the utility is decreasing at these points. Moreover, depending on the relation between the utilities in Equation~\ref{eq:01} (when $p=0$ and $p=1$), two scenario is possible as we illustrate in Figure~\ref{fig:ill}. From this figure it is clear that when $\varphi_u(0)<\varphi_u(1)$ (or $f\cdot(M-in(u))<L$) for all users that the maximum of their utilities cannot be at $p=0$. 
    \hfill\qed
\end{proof}



\begin{figure}[t]
    \centering
    \includegraphics[width=8cm]{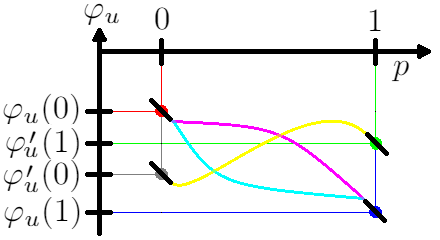}
    \caption{Illustration of the utility functions: the yellow curve's maximum must be between zero and one since the gray dot is below the green where the derivative is negative. }
    \label{fig:ill}
\end{figure}

\paragraph{Potential Game. }
We also show that FMD-RA is a potential game~\cite{monderer1996potential}. 
This is especially important, as it guaranteed that the Best Response Dynamics terminates in a NE. 

\begin{definition}[Potential Game]
A Game $\langle\mathcal{N},\mathcal{A},\mathcal{U}\rangle$ (with players $\{1,\dots,U\}$, actions $\{a_1,\dots,a_U\}$, and utilities $\{\varphi_1,\dots,\varphi_U\}$) is a Potential Game if there exist a potential function $\Psi$ such that Equation \ref{eq:pot} holds for all players $u$ independently of the other player's actions.\footnote{$a_{-u}$ is a common notation to represent all other players action except player $u$. Note that $p(-u)$ stands for the same in relation with FMD. }
\begin{equation}
\label{eq:pot}
    \varphi_u(a_u,a_{-u})-\varphi_u(a_u^\prime,a_{-u})=\Psi(a_u,a_{-u})-\Psi(a_u^\prime,a_{-u})
\end{equation}
\end{definition}

\begin{theorem}
    FMD-RA is a Potential Game with potential function shown in Equation~\ref{eq:pot_FMD}.
    \begin{equation}
    \label{eq:pot_FMD}
    \Psi(p(1),\dots,p(U))=-f\cdot\sum_{u=1}^Up(u)\cdot(M-in(u))
\end{equation}
\end{theorem}

\begin{proof}
    We prove Equation~\ref{eq:pot} by transforming both side to the same form. We start with the left side: the privacy related part of the utility does only depend on the other user's action, therefore this part falls out during subtraction. On the other hand the download related part accumulates as shown below. 
    \begin{align*}
        \varphi_u(p(u),(p(-u))-\varphi_u(p(u)^\prime,p(-u))=\\
        -f\cdot(in(u)+p(u) \cdot(M-in(u)))-(-f\cdot(in(u)+p(u)^\prime \cdot(M-in(u))))=\\
        -f\cdot p(u) \cdot(M-in(u))-(-f\cdot p(u)^\prime \cdot(M-in(u)))
    \end{align*}
    
    Coincidentally, we get the same result if we do the subtraction on the right side using the formula in Equation~\ref{eq:pot_FMD} as all element in the summation besides $u$ falls out (as they are identical because they do not depend on user $u$'s action). 
    \hfill\qed
\end{proof}
\section{Attacks on Privacy}\label{sec:attacks}

We show several possible attacks against the FMD scheme, that might be fruitful to be analyzed in more depth.

\paragraph{Intersection Attacks. }
The untrusted server could possess some background knowledge that it allows to infer that some messages were meant to be received by \emph{the same recipient}. In this case, the server only needs to consider the intersection of the anonymity sets of the ``suspicious'' messages. Suppose the server knows that $l$ messages are sent to the same user. In that case, the probability that a user is in the intersection of all the $l$ messages' anonymity sets is drawn from the $\mathsf{Binom}(U,p^{l})$ distribution. Therefore, the expected size of the anonymity set after an intersection attack is reduced to $p^{l}U$ from $pU$.

\paragraph{Sybil attacks. }
The collusion of multiple nodes would decrease the anonymity set of a message. For instance, when a message is downloaded by $K$ nodes out of $U$, and $N$ node is colluding, then the probability of pinpointing a particular message to a single recipient is $\frac{\binom{N+1}{K}}{\binom{U}{K}}$. This probability clearly increases as more node is being controlled by the adversary.
On the other hand, controlling more nodes does trivially increase the controller's privacy (not message-privacy but user-privacy) as well. However, formal reasoning would require a proper definition for both of these privacy notions.

\paragraph{Neighborhood attacks. }
Neighborhood attacks had been introduced by Zhou et al. in the context of deanonymizing individuals in social networks~\cite{zhou2011k}. An adversary who knows the neighborhood of a victim node could deanonymize the victim even if the whole graph is released anonymously. FMD is susceptible to neighborhood attacks, given that relationship anonymity can be easily broken with statistical tests. More precisely, one can derive first the social graph of FMD users and then launch a neighborhood attack to recover the identity of some users.

\end{document}